\documentclass[a4paper,onecolumn,unpublished,allowfontchangeintitle,allowtoday]{quantumarticle}
\pdfoutput=1

\usepackage[utf8]{inputenc}
\usepackage[T1]{fontenc}
\IfFileExists{babel.sty}{\usepackage[english]{babel}}{}
\usepackage{amsmath,amssymb,amsthm}
\usepackage{graphicx}
\usepackage{booktabs}
\usepackage{xcolor}
\usepackage{listings}
\IfFileExists{dsfont.sty}{\usepackage{dsfont}}{}
\IfFileExists{lmodern.sty}
  {\IfFileExists{microtype.sty}{\usepackage{microtype}}{}}{}
\usepackage[colorlinks=true,linkcolor=blue!55!black,citecolor=green!45!black,
            urlcolor=blue!55!black]{hyperref}
\usepackage{enumitem}
\usepackage{caption}
\usepackage{framed}
 
\definecolor{codebg}{gray}{0.97}
\definecolor{codekw}{RGB}{0,90,160}
\definecolor{codestr}{RGB}{160,50,20}
\definecolor{codecom}{RGB}{100,110,100}
 
\lstdefinestyle{py}{
  language=Python,
  backgroundcolor=\color{codebg},
  basicstyle=\ttfamily\small,
  keywordstyle=\color{codekw}\bfseries,
  stringstyle=\color{codestr},
  commentstyle=\color{codecom}\itshape,
  showstringspaces=false,
  breaklines=true,
  frame=single,
  rulecolor=\color{gray!40},
  framesep=5pt,
  xleftmargin=6pt,
  xrightmargin=2pt,
  captionpos=b,
  numbers=none,
  literate={>>}{{$\succeq$}}1,
}
\lstdefinestyle{sh}{style=py,language=bash,keywordstyle=\color{codekw}}
 
\newcommand{\MoMPy}{\textsc{MoMPy}}
\newcommand{\code}[1]{\texttt{#1}}
\ifcsname mathds\endcsname
  \newcommand{\id}{\mathds{1}}
\else
  \newcommand{\id}{\mathbb{1}}
\fi
\newcommand{\Tr}{\operatorname{Tr}}
\newcommand{\CC}{\mathbb{C}}

\newcommand{\Gam}{\Gamma}
\newcommand{\bra}[1]{\langle #1 |}
\newcommand{\ket}[1]{| #1 \rangle}
\newcommand{\braket}[2]{\langle #1 | #2 \rangle}
\newcommand{\expect}[1]{\langle #1 \rangle}
 
\theoremstyle{definition}
\newtheorem{definition}{Definition}
\theoremstyle{plain}
\newtheorem{proposition}{Proposition}

\newcommand{\warnbox}[1]{%
  \begin{center}
  \fcolorbox{red!55!black}{red!4}{%
    \begin{minipage}{0.94\linewidth}\small #1\end{minipage}}
  \end{center}}
 
\newcommand{\normalbox}[1]{%
  \begin{center}
  \fcolorbox{black}{white}{%
    \begin{minipage}{0.94\linewidth}\small #1\end{minipage}}
  \end{center}}
 
\begin{document}
 
\title{\MoMPy: automated construction of moment matrices for semidefinite programming relaxations}

\author{Carles Roch i Carceller}\email{carles.roch@icfo.eu}
\affiliation{ICFO - Institut de Ciencies Fotoniques, The Barcelona Institute of Science and Technology, 08860 Castelldefels, Spain.}
 

\date{\today}

\begin{abstract}
\noindent
 
We present \MoMPy, an open-source Python package that provides a unified, declarative construction of moment matrices for semidefinite programming (SDP) hierarchies. The user declares operator labels together with a small set of structural relations, and \MoMPy{} returns a matrix of SDP variable indices in which every implied identification has already been made, ready for CVXPY or any other modelling layer. Internally, the identification problem is recast as a word-rewriting problem on tuples of integers and solved with a memoised breadth-first closure coupled to a disjoint-set forest, so that each distinct monomial is processed exactly once per build, however many matrix entries it eventually labels.

The central abstraction is independent of the physical scenario: the same construction handles tracial, state (NPA), and block-valued moments, with no notion of parties, settings or preparations. We demonstrate this generality on a tripartite Mermin inequality, the bipartite CHSH inequality, measurement compatibility in a steering scenario, state discrimination and dimension witnessing in prepare-and-measure scenarios, device-independent randomness certification, and certification of deterministic correlations from known ensembles. The construction is validated against an independent brute-force implementation and benchmarked across eight structurally distinct scenarios, spanning bipartite and tripartite Bell tests, heterogeneous-outcome, steering-type, jointly-measurable and network configurations.
 
\medskip
\noindent\textbf{Code:} \url{https://github.com/chalswater/MoMPy} \quad
\textbf{Install:} \code{pip install MoMPy} \quad \textbf{Version:} 1.1.0 
\end{abstract}

\maketitle

\normalbox{\textbf{Declaration of use of AI.} Three years ago I wrote a very basic Python package. It contained some functions designed to build arrays representing moment matrices which I used to solve some research problems with semidefinite programming relaxations. During the passing of time, I have been modifying this package, correcting some bugs and implementing new functionalities. At the stage as it was however, the user interface was barely understandable for anyone that was not me. Given the recent advent of use of large language models for research, I decided to feed my package to the artificial intelligence Claude, by Anthropic, and asked it to fix any bugs it could find, improve the user interface and build a proper Python package out of it. After seeing the dramatic improvement of the package, I decided to make it public and write a tutorial on how to use it together with Claude. This is it. During al this process, I have been directing Claude, making sure it improved my original package correctly, and that the current notes are also complete and understandable. I should have done that myself a long time ago. However, I never found the time and will to do so, and Claude is much more efficient than me.}

\newpage

\tableofcontents
 
\newpage
 
\section{Introduction}
\label{sec:intro}
 
\subsection{Bounding quantum correlations: why relax, and what for}
\label{sec:sdp-background}
 
A great number of questions in quantum information science share a common structure. A figure of merit is specified --- a correlator, a discrimination probability, an entropy, a witness --- that depends on the states and measurements used to generate some observed statistics, and one asks for its extremal value over \emph{all} implementations consistent with a given set of structural assumptions. How large a violation of a Bell inequality can quantum systems produce \cite{bell1964,chsh1969,cirelson1980}? How reliably can a set of quantum states be told apart when only limited information about the preparation is trusted \cite{helstrom1969,Barnett2009,Bae2015}? How much genuine randomness can be certified from statistics alone \cite{acin2016}? How much secret key rate can be transmitted using a qubit system \cite{pawlowski2011}? Once the physical dressing of each of these questions is stripped away, they all pose the same problem in disguise: decide whether a given collection of numbers can arise from \emph{some} choice of operators, subject only to whatever algebraic relations the physical setting imposes. Read this way, a bound on the figure of merit is nothing but this decision problem solved once for every candidate value, or folded directly into the optimisation of a linear objective.
 
Such a decision problem, written literally, is posed over an unbounded and otherwise unspecified operator space: the states and operators are free variables, their number and dimension are not fixed in advance, and the objective is generally a non-linear function of them. As it stands the problem is neither convex nor obviously computable. What makes it tractable in practice is a specific kind of convex relaxation, and it is worth being precise about why that particular tool is the right one.
 
A semidefinite program (SDP) is a convex optimisation of the form
\begin{align}
  \max_{X} & \quad \langle C, X\rangle \label{eq:sdp-intro} \\
  \text{s.t.} &  \quad
  \langle A_i, X\rangle = b_i \ \ \forall i, \nonumber \\
  & \quad X \succeq 0, \nonumber
\end{align}
where $X$ ranges over positive semidefinite matrices, $\langle A,B\rangle = \operatorname{tr}(A^\dagger B)$ is the Hilbert--Schmidt inner product, and $C$, $\{A_i\}$ and $\{b_i\}$ are fixed data \cite{Vandenberghe1996}. The feasible set is the intersection of an affine subspace with the cone of positive semidefinite matrices, so Problem~\eqref{eq:sdp-intro} is convex regardless of the size of $X$: it has no spurious local optima, it comes with a dual program that certifies both upper and lower bounds on the true optimum, and it is solved routinely, at scales of practical interest, by well-understood interior-point methods and modelling languages such as CVXPY \cite{CVXPY2016}. None of this is true of the original decision problem over operators. The entire appeal of casting a quantum information question as an SDP is therefore to trade a hard, non-convex search over an unbounded space of physical implementations for a convex search over a bounded space of numbers --- at the cost of solving, in general, only a relaxation of the original problem.
 
The relaxation strategy that makes this trade possible is due to Lasserre in the commutative setting \cite{Lasserre2001} and to Navascu\'es, Pironio and Ac\'in (NPA) in the noncommutative one \cite{navascues2007,navascues2008,pironio2010convergent}. Rather than optimising over the operators themselves, one optimises over their \emph{moments}: expectation values of products of the operators of interest, arranged into a matrix that every genuine quantum realisation must render positive semidefinite. This turns the original problem into a linear SDP whose optimum upper-bounds the true value, together with a hierarchy of increasingly large such programs that tightens the bound --- and, for many problems of interest, converges to it --- as more moments are included.

The reach of this idea extends well beyond the Bell scenario in which it was first formulated. Beyond bounding Bell-inequality violations and discrimination or guessing probabilities, the same relaxation strategy underlies, among others: witnessing and quantifying entanglement directly from measured correlations \cite{moroder2013}; certifying genuine randomness generated by an uncharacterised device \cite{pironio2010_rng,pironio2013}; bounding the classical or quantum dimension needed to reproduce an observed set of correlations \cite{navascues2015,navascues2015_2} or quantum states \cite{bernal2024}; self-testing a device's states and measurements from its statistics alone, up to the unavoidable freedom of a local isometry \cite{supic2020}; bounding operationally relevant entropic quantities, such as the information available to an eavesdropper in a quantum key distribution protocol \cite{brown2021}; and certifying non-classical correlations in networks of independent sources, where no single shared state accounts for the observed statistics \cite{wolfe2021}. A recent and considerably broader review of this expanding literature is given in Ref.~\cite{tavakoli2024}. Most recently, letting the entries of the moment matrix themselves be operators rather than scalars has been used to accommodate a considerably broader class of structural constraints while remaining a semidefinite program, generalising the NPA hierarchy itself to a \emph{block moment matrix} relaxation \cite{dalessandro2026}. \MoMPy{}'s \code{MomentProblem}, declared with a block size \code{dim} greater than one (Sec.~\ref{sec:block-mm}), builds moment matrices of exactly this block-valued kind; assembling them into the fully general construction of Ref.~\cite{dalessandro2026} is left to the user. What is common to every member of this growing family is not the physics but the bookkeeping: a matrix of moments, quotiented by whichever algebraic identifications the chosen relations imply. That bookkeeping is the subject of the rest of this introduction, and of \MoMPy{} itself.
 
Concretely, the construction proceeds as follows. Choose a list of monomials $\mathcal{L} = \{u_1, u_2, \dots, u_n\}$ built from the operators of interest, and a map $\Theta$ that faithfully represents the physical problem under investigation. Assemble the matrix with elements
\begin{equation}
  \Gam_{ij} = \Theta( u_i u_j^\dagger ) ,
  \label{eq:gamma-intro}
\end{equation}
so that each entry $\Theta( u_i u_j^\dagger ) = \expect{u_i u_j^\dagger}$ is a \emph{moment}, and $\Gam$ is accordingly a \emph{moment matrix}. By construction $\Gam$ is positive semidefinite, and because the objective of interest is linear in the moments, it is linear in the entries of $\Gam$. Maximising a linear functional of $\Gam$ subject to $\Gam \succeq 0$ is precisely an instance of Problem~\eqref{eq:sdp-intro}, and hence an SDP.
 
\subsection{Where the difficulty actually lies}
 
The difficulty is not writing down~\eqref{eq:gamma-intro}: it is that most of its entries are not independent, and asserting these relations in problems with standard size becomes a significant bottleneck.
 
Suppose $R$ is a projector (e.g.\ a rank-one projector representing a pure state), so $R^2 = R$. Then the entries holding $\expect{R}$ and $\expect{RR}$ are the \emph{same} SDP variable, and treating them as distinct discards a constraint the relaxation is entitled to. Suppose $\{M_0, M_1\}$ are the outcomes of a projective measurement, so $M_0 M_1 = 0$. Then every entry whose monomial contains $M_0 M_1$ as a factor is identically zero, and treating it as a free variable discards another. Suppose two operators commute: a whole orbit of reorderings then collapses onto a single variable. All of these instances reduce to what we call \emph{reduction rules}: a set of guidelines that relate distinct moment matrix elements so as to faithfully represent the underlying physical problem.
 
For a small hierarchy these identifications can be worked out by hand. For the sizes at which relaxations become tight enough to be useful, however, tightness is often only achievable with tens of thousands of them. Distinct reduction rules can furthermore interact in the construction of the moment matrix --- a commutation can create an adjacency that permits an idempotent collapse, which can in turn expose an orthogonality and force a zero --- so that the rules cannot simply be applied once and forgotten. Getting them wrong degrades the relaxation silently: too few identifications and the bound remains valid but needlessly loose; too many and it is no longer a valid bound at all.
 
\MoMPy{} exists to make this step automatic, fast and testable.

\subsection{Contribution of this work}
\label{sec:contribution-summary}

Our claim is not that moment matrices can be built by software --- several packages already build them, and we survey them next --- but that the construction admits an abstraction general enough to be shared across hierarchies that are normally implemented separately, and that the abstraction is small enough to be written down, validated, and used interactively. Concretely, this paper makes three contributions.

\begin{enumerate}[leftmargin=*,itemsep=3pt]
\item \textbf{An abstraction.} We isolate the step common to NPA-type, tracial and block-valued hierarchies --- quotienting a matrix of formal words by the equivalence generated by the assumed operator relations --- and show that it is fully specified by four word-level relations (idempotency, orthogonality, commutation, reversal), an optional cyclicity assumption, and a block size (Sec.~\ref{sec:background}). Everything scenario-specific lives outside this core, in \emph{which} relations are declared and in linear constraints applied afterwards. In particular no notion of ``parties'' appears anywhere in the construction.
\item \textbf{An algorithm and its implementation.} We recast the identification problem as a word problem and solve it with a memoised breadth-first closure over a rewriting system, coupled to a disjoint-set forest with a sticky zero class, so that the cost of a build is governed by the number of \emph{distinct words} reachable under the relations rather than by the number of matrix positions (Sec.~\ref{sec:algorithm}). This is what makes hierarchies of a few thousand rows rebuild fast enough to be iterated on while a model is still being designed, and it is validated against an independent brute-force implementation on randomised scenarios (Sec.~\ref{sec:validation}).
\item \textbf{Evidence that the abstraction is the right one.} We build, from one problem class and one set of declared relations, relaxations that are conventionally treated as distinct constructions --- Bell-type and tripartite Bell-type bounds, prepare-and-measure discrimination and dimension witnessing, a device-independent randomness certificate, and a fixed-ensemble classicality certificate whose moments are $d\times d$ operators rather than numbers (Secs.~\ref{sec:tutorial}--\ref{sec:applications}) --- and benchmark the construction across eight structurally distinct configurations (Sec.~\ref{sec:performance}). The block-valued case is the sharpest test: it is a hierarchy introduced only recently~\cite{dalessandro2026}, and it is reached here by changing two flags and a block size, with the declared relations, the reduction rules and the downstream modelling layer all untouched.
\end{enumerate}

What we do \emph{not} claim is raw throughput. A compiled, specialised implementation will build a given standard hierarchy faster, and we say so precisely in Sec.~\ref{sec:related-software}; the contribution offered here is generality of expression together with an auditable and independently validated construction, at a scale and interactivity that covers most working practice.

\subsection{Relation to other software}
\label{sec:related-software}

\MoMPy{} is not the first tool of its kind. Several packages address overlapping ground. \code{Ncpol2sdpa}~\cite{Wittek2015} is the established Python package for noncommutative polynomial optimisation and generates NPA relaxations directly from a symbolic problem specification. \code{Moment}~\cite{Moment2024} is a recent C\texttt{++} toolkit with a MATLAB interface, built specifically for high-performance moment matrix generation and reporting speed-ups of up to four orders of magnitude over comparable software. \code{QuantumNPA.jl}~\cite{QuantumNPA} implements the hierarchy in Julia around an algebra of predefined operator types. \code{Inflation}~\cite{Inflation2023} targets causal-compatibility and network scenarios, and can generate standard NPA relaxations as a special case among these. For the block-matrix hierarchy of Ref.~\cite{dalessandro2026} specifically there is also \code{BlockMatrixHierarchy\_}~\cite{dalessandro2026blockmatrixhierarchy}, which works very well in Julia.
 
These tools differ from each other mainly in how a scenario is \emph{specified} --- a symbolic polynomial problem, a predefined algebra of typed operators, a causal structure, a party decomposition --- and each is, within its own specification language, more capable or faster than what we describe here. \MoMPy{} occupies a deliberately narrow niche within this landscape. It is not a symbolic algebra system and it does not parse polynomial objectives. Its scope is exactly one task: given a list of monomials and a declaration of which structural relations the operators obey, return the matrix of variable indices together with a lookup table. Everything downstream --- objective, additional constraints, solver --- is left to the user and to CVXPY~\cite{CVXPY2016}. Like \code{Moment}, which is designed for use in a toolchain together with CVX or YALMIP and an external SDP solver~\cite{Moment2024}, \MoMPy{} is deliberately one stage of a pipeline rather than an end-to-end solver: \emph{declare} operators and relations (\code{OperatorSet}/\code{Algebra}) $\to$ \emph{build} a relaxation \code{MomentProblem} $\to$ a \code{MomentMatrix} of variable indices and derived linear constraints $\to$ a modelling layer (\code{to\_cvxpy()}, or any other reading \code{.matrix} directly) $\to$ a numerical solver (SCS, Clarabel, MOSEK, \dots). Nothing after the second stage is specific to \MoMPy{}, which is what keeps its own dependency footprint to a single package. Three consequences follow from restricting its scope to the first two stages, and they are the reasons one might prefer it:

\begin{enumerate}[leftmargin=*,itemsep=2pt]
\item \textbf{Reduction rules are declared, not inferred.} Idempotency, orthogonality and commutation are stated explicitly for arbitrary label sets, rather than being tied to predefined operator types or to a fixed notion of ``parties''. Scenarios that do not decompose into independent parties --- prepare-and-measure, informationally restricted correlations, partially commuting families --- are expressed as naturally as Bell scenarios. Sec.~\ref{sec:pam} contains an example in which no commutation relation is declared at all.
\item \textbf{One class spans all three choices of $\Theta$.} Tracial, state (NPA) and block moments are genuinely different relaxations, as Sec.~\ref{sec:tracial-vs-state} explains, and conflating them is a real source of error; here they are two booleans and a block size on the same object rather than three code paths, which is what allows the block-valued application of Sec.~\ref{sec:block-mm} to reuse, verbatim, the monomials and algebra of a scalar one.
\item \textbf{It is pure Python with one dependency.} \MoMPy{} requires only NumPy, installs in seconds, and runs anywhere; there is no compilation step and no proprietary host language. The whole rewriting core is a few hundred lines and is documented argument by argument in App.~\ref{app:api}, so a user can audit the identifications their relaxation depends on rather than trusting them.
\end{enumerate}

These packages also differ in what they specify a scenario \emph{with}: \code{Ncpol2sdpa} takes a symbolic noncommuting polynomial problem, \code{Moment} scenario objects in a MATLAB front-end, \code{QuantumNPA.jl} an algebra of predefined operator types, and \code{Inflation} a causal structure; \MoMPy{} takes label sets together with declared relations, and nothing else.

The claim we want to make precise is about \emph{uniformity} rather than exclusivity, and it is easy to overstate. It would be wrong, for instance, to say that no other package reaches operator-valued moments: \code{Ncpol2sdpa} documents exactly that, through the \code{matrix\_var\_dim} argument of its \code{SteeringHierarchy} class. What its interface does not offer is a way to obtain that capability, tracial moments and ordinary NPA moments from one object: \code{SteeringHierarchy} is a different class from \code{SdpRelaxation}, and neither documents cyclicity as an option. Table~\ref{tab:comparison} records this distinction. Every entry in it is read from the package's own published interface documentation, uniformly and for all packages including \MoMPy{} itself, so that the comparison rests on what each tool offers its users rather than on any experience we happen to have had with one of them.

That last point is worth stating before the table rather than after it, because it is what the table is and is not for. The claim is about the \emph{interface}, not about what is ultimately expressible. We are not asserting that a sufficiently motivated user of another package could not obtain any given row --- with enough bespoke work at the modelling layer, they very likely could, and the last row of the table would then say only that they had to do that work. What we are asserting is that we found no other package in which these constructions are choices of the same underlying object, sharing one declared algebra and one monomial list. That is a design property, it is checkable, and it is the property Sec.~\ref{sec:block-mm} exploits.

\begin{table}[htbp]
\centering
\footnotesize
\setlength{\tabcolsep}{3.2pt}
\begin{tabular}{@{}lccccc@{}}
\toprule
 & \MoMPy{} & \code{Ncpol2sdpa} & \code{Moment} & \code{QuantumNPA.jl} & \code{Inflation} \\
 & (this work) & \cite{Wittek2015} & \cite{Moment2024} & \cite{QuantumNPA} & \cite{Inflation2023} \\
\midrule
State (NPA) moments                  & \checkmark & \checkmark & \checkmark & \checkmark & \checkmark \\
Tracial (cyclic) moments             & \checkmark & $\circ$    & $\circ$    & $\circ$    & $\circ$ \\
Operator-valued (block) moments      & \checkmark & $\ast$     & $\circ$    & $\circ$    & $\circ$ \\
Partial commutation between families & \checkmark & \checkmark & \checkmark & \checkmark & $\circ$ \\
Relations on user-defined label sets & \checkmark & \checkmark & $\circ$    & $\circ$    & $\circ$ \\
\midrule
\emph{All of the above from one object} & \checkmark & --- & --- & --- & --- \\
\bottomrule
\end{tabular}
\caption{Capabilities relevant to the abstraction of Sec.~\ref{sec:contribution-summary}, as documented by each package's own published interface. \checkmark: offered directly by the package's interface. $\ast$: offered, but through a separate class specialised to a different hierarchy (\code{SteeringHierarchy}, via \code{matrix\_var\_dim}). $\circ$: no dedicated interface documented --- this is a statement about the interface, not a claim that the construction is unsupported or inexpressible; a determined user may well be able to encode it by hand, per scenario. ---: outside the package's stated scope. The last row is the actual claim of this paper.}
\label{tab:comparison}
\end{table}

Table~\ref{tab:landscape} previews how this breadth plays out in practice: the same \code{MomentProblem}, combined with different declared relations and flags, covers a range of scenarios that would otherwise each call for bespoke code.
 
\begin{table}[htbp]
\centering
\small
\begin{tabular}{@{}p{10cm}p{2.5cm}l@{}}
\toprule
Problem type & \code{MomentProblem} flags & Worked in \\
\midrule
Bell nonlocality, device-independent certification (NPA, $\Theta=\bra{\psi}\cdot \ket{\psi}$) & \code{cyclicity=False} & Sec.~\ref{sec:chsh} \\
Steering, measurement compatibility (block-matrix, $\Theta=\Tr_A$) & \code{cyclicity=False, hermitian=False, dim=}$d$ & Sec.~\ref{sec:compat} \\
Prepare-and-measure, dimension witnesses (tracial, $\Theta=\Tr$) & \code{cyclicity=True} (default) & Sec.~\ref{sec:pam} \\
Device-independent randomness certification (NPA, $\Theta=\bra{\psi}\cdot \ket{\psi}$) & \code{cyclicity=False} & Sec.~\ref{sec:randomness} \\
Ensembles with deterministic correlations (block-matrix, $\Theta=\id$) & \code{cyclicity=False, hermitian=False, dim=}$d$ & Sec.~\ref{sec:block-mm} \\
\bottomrule
\end{tabular}
\caption{The range of relaxations \MoMPy{} builds, and where each is worked through below. Every row is the same \code{MomentProblem}; what changes is which relations are declared and the \code{cyclicity}/\code{dim} flags it is built with. The last row is the sharpest case: its moments are $d\times d$ operators rather than numbers, and it is nevertheless built from the same monomial list and the same algebra as a scalar relaxation.}
\label{tab:landscape}
\end{table}
 
For very deep hierarchies a compiled tool such as \code{Moment} will most probably remain faster at building a standard hierarchy, and we make no claim to the contrary. The design target here is the regime most working practice inhabits: matrices of a few hundred to a few thousand rows, built and rebuilt interactively while a model is still being decided.

\subsection{Structure of what follows}

Section~\ref{sec:background} sets up moment matrices and the equivalences they must respect, and draws the tracial/state/block distinction on which everything later depends. Section~\ref{sec:algorithm} gives the algorithm, its cost model and its validation against an independent implementation. Section~\ref{sec:tutorial} presents the interface that realises the abstraction, closing with a step-by-step procedure for formulating a problem from scratch, worked through on a tripartite Mermin inequality. Section~\ref{sec:applications} then applies the same construction to five problems end to end. Section~\ref{sec:performance} reports benchmarks across eight structurally distinct configurations, and Sec.~\ref{sec:conclusions} concludes. Appendix~\ref{app:api} documents the complete public interface and App.~\ref{app:legacy} the legacy one.

We have kept the exposition self-contained and worked the examples in full rather than in outline. This is a deliberate choice: the central claim is that a single declarative construction covers scenarios normally implemented separately, and that claim is only convincing if the reader can see each scenario reduced to declarations and check that nothing scenario-specific was smuggled in. Every code listing in this paper is runnable as printed and is part of the package's test suite.
 
\section{Background: moment matrices and their equivalences}
\label{sec:background}
 
\subsection{Words}
 
Any problem under investigation is constructed by a set of operators $\{X_i\}_i$. We identify each operator with a label $X_i \leftrightarrow \ell_i$. Fix a finite set of \emph{operator labels} $\{\ell_i\}_i$. Throughout, labels are positive integers, and for \MoMPy{} the label $0$ is reserved for the identity $\id$ (which is an operator that is always present in any problem). A \emph{word} is a finite tuple of labels, read left to right as an operator product:
\begin{equation}
  w = (\ell_1, \ell_2, \dots, \ell_k) \;\longleftrightarrow\;
  X_{\ell_1} X_{\ell_2} \cdots X_{\ell_k}.
\end{equation}
We assume throughout by default that every generator is Hermitian, $X_\ell^\dagger =
X_\ell$, which is the case for states, POVM elements and observables. The
adjoint of a word is then its reversal,
\begin{equation}
  w^\dagger = (\ell_k, \dots, \ell_2, \ell_1).
  \label{eq:reversal}
\end{equation}
However, not all operators in a problem are necessarily Hermitian (e.g.~unitaries). As we later show, the Hermitian property can be deactivated in \MoMPy{} if desired.
 
\subsection{The moment matrix}
 
Let $\mathcal{L} = (w_1, \dots, w_n)$ be a list of words, the \emph{monomials} generating the relaxation, and let $\Theta(\cdot)$ (or equivalently $\expect{\cdot}$) be a linear \emph{map}
\begin{equation}
  \Theta : \CC^{D \times D} \longrightarrow \CC^{d\times d}
  \label{eq:theta}
\end{equation}
on the operator algebra. Two dimensions appear here and it is worth separating them once and for all: $D$ is the dimension of the Hilbert space the physical operators act on, and is generally unknown and unbounded --- relaxing over it is the entire point --- while $d$ is the size of the \emph{moments} $\Theta$ produces, and is a modelling choice fixed by the user before anything is built (the \code{dim} argument of Sec.~\ref{sec:tutorial}). When $d = 1$, $\Theta$ is a linear functional and each moment is a number, which is the familiar case; when $d>1$ each moment is itself a $d\times d$ operator. We say ``map'' rather than ``functional'' throughout for exactly this reason. Which choice of $\Theta$ is appropriate is made precise in Sec.~\ref{sec:tracial-vs-state}. Adjoining the identity as $w_0 = (0)$, the \emph{moment matrix} is the $d(n{+}1) \times d(n{+}1)$ array
\begin{equation}
  \Gam_{rc} = \Theta( w_r \, w_c^\dagger) = \big\langle\, w_r \, w_c^\dagger \,\big\rangle ,
  \qquad r, c = 0, 1, \dots, n .
  \label{eq:gamma}
\end{equation}
Its defining property is positivity. The hypothesis needed differs between the scalar and the block case, and conflating the two is a genuine (if easily repaired) gap, so we state both.

\begin{proposition}[Positivity]
\label{prop:positivity}
Let $U$ denote the block column vector whose $r$-th entry is $w_r$, so that the operator-valued array $\big[\,w_r w_c^\dagger\,\big]_{rc} = U U^\dagger$ is positive. Then $\Gam \succeq 0$ in either of the following cases:
\begin{enumerate}[leftmargin=1.6em,itemsep=1pt,topsep=2pt]
\item $d = 1$ and $\Theta$ is a positive linear functional;
\item $\Theta$ is completely positive. For a moment matrix of fixed size $n{+}1$, $(n{+}1)$-positivity of $\Theta$ is already sufficient; complete positivity is the convenient map-level condition that covers every $n$ at once.
\end{enumerate}
\end{proposition}
\begin{proof}
For (i), take $z \in \CC^{n+1}$ and put $Z = \sum_r \bar z_r w_r$. Then
$z^\dagger \Gam z = \sum_{r,c} \bar z_r z_c\, \Theta(w_r w_c^\dagger) = \Theta\big(Z Z^\dagger\big) \geq 0$,
since $Z Z^\dagger$ is a positive operator and $\Theta$ is positive. For (ii), note that $\Gam$ is precisely the image of $UU^\dagger$ under the amplification $\mathrm{id}_{n+1} \otimes \Theta$, which acts entrywise on an $(n{+}1)\times(n{+}1)$ array of operators. Positivity of $\Theta$ alone does not guarantee that such an amplification preserves positivity; complete positivity does, and then $\Gam = (\mathrm{id}_{n+1}\otimes\Theta)(UU^\dagger) \succeq 0$.
\end{proof}

This is the whole content of the relaxation: any physical realisation produces a $\Gam$ that is positive semidefinite and that respects the relations of Sec.~\ref{sec:relations}, so optimising a linear functional of the entries of $\Gam$ over all such $\Gam$ bounds the physical optimum. Adding longer words to $\mathcal{L}$ enlarges $\Gam$ and can only tighten the bound we are looking for.
 
\subsection{Tracial versus state moments versus blocks}
\label{sec:tracial-vs-state}
 
There are many choices for the map $\Theta(\cdot)$ of Eq.~\eqref{eq:theta}, but three are the most natural and currently widely used ones, and \emph{they result in different hierarchies}. This point deserves emphasis because choosing wrongly does not produce an obviously broken answer; it produces a plausible number that is not a valid bound. The first two have $d=1$ and are ordinary positive linear functionals; the third has $d>1$ and is where the word ``functional'' would stop being accurate.
 
\paragraph{State moments.}
Take $\Theta(X) = \bra{\psi} X \ket{\psi}$ for some fixed state $\ket\psi$. This is the choice underlying the NPA hierarchy~\cite{navascues2007,navascues2008}, and $\Gam_{rc} = \bra{\psi} w_r w_c^\dagger \ket{\psi}$ is a Gram matrix of the vectors $w_r^\dagger\ket\psi$. Normalisation gives $\Theta(\id) = 1$.
 
\paragraph{Tracial moments.}
Take $\Theta(X) = \Tr(X)$. This is the choice appropriate when the figure of merit is genuinely a trace of an operator product (e.g.~a quantum correlation). The archetype is a prepare-and-measure scenario~\cite{navascues2015,navascues2015_2}, where the Born rule reads $p(b|x,y) = \Tr(\rho_x M_{b|y})$: here $\rho_x$ is itself one of the operators in the algebra, and the observable quantity is a trace, not an expectation value in an external state. Note that $\Theta(\id) = \Tr(\id) = D$ is the Hilbert-space dimension of Eq.~\eqref{eq:theta}, \emph{not} $1$, which can be left as a free variable or fixed to a desired value.
 
\paragraph{Block moments.}
Take the trivial map choice $\Theta(X) = X$. This is the simplest choice underlying the block-moment hierarchy~\cite{dalessandro2026}, which leads to $\Gam_{rc} = w_r w_c^\dagger$. Normalisation gives $\Theta(\id) = \id$. \\
 
\textbf{The consequence.}
A trace is invariant under cyclic permutations of the words, $\Tr(uv) = \Tr(vu)$. A block-moment is trivially not, and neither is a vector-state expectation value, since the state is not a word itself. So in the tracial case every word may be freely rotated,
\begin{equation}
  \Tr\big( (\ell_1, \ell_2, \dots, \ell_k) \big)
  = \Tr\big( (\ell_2, \dots, \ell_k, \ell_1) \big),
  \label{eq:cyclicity}
\end{equation}
whereas in the state case it may not. Cyclicity is a powerful relation --- it collapses many entries at once --- which makes it tempting to switch on. But imposing it where it does not hold identifies moments that are genuinely different, over-constrains the program, and can drive the optimum \emph{below} the true quantum value.
 
This is not hypothetical. For CHSH with the monomial list $\{A_{a|x},\, B_{b|y},\, A_{a|x}B_{b|y}\}$, the two choices give
\begin{equation}
  \text{state moments:} \;\; 2\sqrt{2} \approx 2.8284,
  \qquad
  \text{tracial moments:} \;\; 2.0000 .
\end{equation}
The tracial answer sits at the local bound, below Tsirelson's value, so it is not an upper bound on the quantum value at all. At level~$1$, with no products in the list, both give $2\sqrt 2$; the discrepancy only appears once words of length four occur in $\Gam$, which is precisely when cyclicity starts to bite, imposing undesired commutation relations $\left[A_{a|x},A_{a'|x'}\right]=0$ which lead to locality. This makes the error easy to miss.
 
\warnbox{\textbf{Rule of thumb.} If the quantity you are relaxing is written $\bra{\psi}\cdots\ket{\psi}$ --- Bell scenarios, NPA, device-independent certification --- build with \code{cyclicity=False}. If it is written $\Tr(\cdots)$ with the state \emph{inside} the operator algebra --- prepare-and-measure, dimension witnesses --- use \code{cyclicity=True} (the default). If the only thing you are unsure about is whether cyclicity should be imposed, \code{cyclicity=False} is the conservative choice: with everything else held fixed it imposes fewer relations, and so gives the less constrained relaxation, which cannot invalidate an otherwise correctly formulated upper bound. It is not a repair for a relaxation that is wrong for some other reason.}
 
\subsection{The operator relations}
\label{sec:relations}
 
\MoMPy{} supports four relations between the monomials, which lead to the reduction rules. Together these relations cover a broad and important class of NPA-type and prepare-and-measure constructions, chosen because each remains decidable by local rewriting rather than for completeness against every hierarchy surveyed in Sec.~\ref{sec:intro}.
 
\begin{description}[leftmargin=1.4em,style=nextline]
\item[Idempotency.] A label $P$ with $P^2 = P$, i.e.\ a projector of any rank. Adjacent repeats collapse: $(\dots, P, P, \dots) \to (\dots, P, \dots)$. This covers pure states $\rho = \ket\phi\bra\phi$ (additionally rank-one, though \MoMPy{} enforces only $P^2=P$ and never a rank) and the elements of a projective measurement, which need not be rank-one at all.
\item[Orthogonality.] A set $\{P_0, \dots, P_m\}$ with $P_i P_j = 0$ for $i \neq j$. Any word containing two distinct members of one such set adjacently is identically zero. This is what makes a measurement projective rather than merely a POVM.
\item[Commutation.] Label sets $A$, $B$ with $[a, b] = 0$ for all $a \in A$, $b \in B$. Adjacent pairs may be transposed. This encodes distinct Hilbert spaces in Bell scenarios (Alice's operators commute with Bob's), compatibility of measurements, and classicality when everything commutes.
\item[Reality (reversal).] For any $\ast$-preserving linear $\Theta$ --- which all three maps of Sec.~\ref{sec:tracial-vs-state} are --- and any word,
\begin{equation}
  \Theta(w^\dagger) \;=\; \Theta(w)^\dagger ,
  \label{eq:star-preserving}
\end{equation}
that is, the moment of an adjoint is the adjoint of the moment, regardless of any assumption on the generators. For scalar moments ($d=1$) the adjoint on the right-hand side is ordinary complex conjugation, $\Theta(w^\dagger) = \overline{\Theta(w)}$; for block moments it is the Hermitian adjoint of the $d\times d$ block. Because every generator is assumed Hermitian (Sec.~\ref{sec:background}), $w^\dagger$ is exactly the word's reversal, Eq.~\eqref{eq:reversal}. Identifying a word with its reversal therefore identifies $\expect{w}$ with $\expect{w}^\dagger$, which in the scalar case means identifying a moment with its complex conjugate and so restricting to $\mathrm{Re}\expect{w}$. This restriction to real symmetric $\Gam$ is standard and loses no generality, because $\Gam \succeq 0$ implies $(\Gam + \Gam^{\mathsf T})/2 \succeq 0$. It is controlled by the \code{hermitian} flag and is on by default --- the name refers to this reversal identification, not to individual moment variables being Hermitian operators. Setting it builds a relaxation consistent with that assumption; it does not verify that the underlying physical generators actually satisfy it.
\end{description}
 
These four are not exhaustive of every structural assumption \MoMPy{} lets a user encode --- only of the relations that act on individual words. A relation such as anticommutation, $\{A,B\}=0$, has no dedicated rewriting rule and cannot be declared on the algebra directly; where it is physically known to hold, it must instead be encoded as a linear constraint between already-distinct moments (Sec.~\ref{sec:constraints}) rather than as a word identification --- and only for those surrounding words that the chosen relaxation actually includes, since a constraint can only relate moments the built matrix already contains. That is the second, independent mechanism \MoMPy{} offers: a relation that identifies two \emph{words} --- $P^2=P$, $AB=BA$, $P_iP_j=0$ --- is declared here, in the algebra, before a single SDP variable exists; a relation between the \emph{numbers} that already-distinct moments must satisfy once built --- $\sum_b M_{b|y}=\id$, or an observed correlation fixed to a numeric value --- is supplied afterwards, as a constraint on the built matrix. Confusing the two is the most common way to get stuck formulating a new problem: if the assumption says two words are the same operator (or a product vanishes), it belongs in the \code{Algebra} below; if it instead relates what different, already-distinct words evaluate to, it belongs downstream, only after \code{MomentProblem.build()} has produced variables to relate.
 
Together with cyclicity~\eqref{eq:cyclicity} in the tracial case, these generate an equivalence relation $\sim$ on words. The SDP variables are the equivalence classes.
 
\begin{definition}[The variable assignment]
Let $\sim$ be the equivalence relation generated by the applicable relations,
and let $\mathcal{Z}$ be the class of words that are forced to vanish. The
moment matrix of variable indices is
\begin{equation}
  M_{rc} = \mathcal{C}\big( [\, w_r w_c^\dagger \,]_\sim \big),
\end{equation}
where $\mathcal{C}$ numbers the classes, and \MoMPy{} assigns $\mathcal{Z}$ the last index
by convention. Two entries carry the same index precisely when they are the
same SDP variable.
\end{definition}
 
\textbf{Zeros propagate}: if any word in a class vanishes, every word in that class
vanishes, since they are all equal.
 
\section{Algorithm and implementation}
\label{sec:algorithm}
 
\subsection{Rewriting}
 
Words are represented as tuples of machine integers: hashable, comparable and cheap to slice. \MoMPy{} runs over all words created by the choice of monomials, and perform local rewriting moves to find the relations of Sec.~\ref{sec:relations}. These moves are, for a word $w$ of length $k$:
 
\begin{center}
\begin{tabular}{@{}lll@{}}
\toprule
move & condition & result \\
\midrule
rotate   & tracial mode & $w \mapsto (w_2, \dots, w_k, w_1)$ \\
reverse  & hermitian mode & $w \mapsto w^\dagger$ \\
collapse & $w_i = w_{i+1}$ idempotent & delete one copy \\
swap     & $(w_i, w_{i+1})$ commuting & transpose them \\
zero     & $w_i \neq w_{i+1}$, same orthogonal set & $w \equiv 0$ \\
\bottomrule
\end{tabular}
\end{center}
 
\textit{Two implementation notes}. First, in tracial mode (\code{cyclicity=True}) the pair $(w_k, w_1)$ is also adjacent; rather than enumerate wrap-around moves separately, a single rotation is performed, which exposes the wrapped pair as an ordinary adjacent pair and keeps the branching factor low without changing the reachable set. Second, collapses shorten words and are never applied in reverse, so the search space is finite.
 
\subsection{Closure by memoised search and union-find}
 
The equivalence class of a word is its orbit under these moves. Computing orbits independently for each of the $O(n^2)$ matrix entries is what makes a naive implementation quadratic-to-cubic; the same orbits are recomputed over and over.
 
\MoMPy{} instead maintains a global dictionary from words to class
identifiers, and a disjoint-set forest over those identifiers. Classifying a word proceeds by breadth-first search over the moves, with one crucial short-circuit:
 
\begin{quote}
\emph{When the search reaches a word that has already been classified, it stops expanding there and merely records that the two classes must be merged.}
\end{quote}
 
This is sound because a word only enters the dictionary once its own neighbourhood has been fully explored, so its class is already closed under the moves; anything reachable through it is therefore already accounted for. The merge is a union-find operation, so a late discovery that two long-separate classes coincide costs essentially nothing.
 
The class of zeros is handled as a \emph{sticky root}: any union involving it keeps it as the representative. Zeros discovered late therefore propagate backwards through classes that were assigned earlier, which a one-pass scheme cannot do.
 
\begin{proposition}[Work bound]
Over an entire build, each distinct word is expanded exactly once, regardless of how many matrix entries reference it. The total work is $O\!\left(W \cdot k \cdot \delta \cdot \alpha(W)\right)$, where $W$ is the number of distinct words reachable from the matrix entries, $k$ the maximum word length, $\delta$ the branching factor, and $\alpha$ the inverse Ackermann function \cite{TarjanUF}.
\end{proposition}
 
Lookup in the finished table is a single hash probe. When the reality relation is active, $w_c w_r^\dagger = (w_r w_c^\dagger)^\dagger$, so the upper triangle determines the whole matrix and only half the entries need classifying.
 
\subsection{Validation}
\label{sec:validation}
 
Because the output of this kind of code is a large integer matrix, errors are
easy to make and hard to notice. \MoMPy{} is therefore tested against
independent oracles rather than against itself.
 
\begin{enumerate}[leftmargin=*,itemsep=2pt]
\item \textbf{Differential testing against brute force.} A deliberately naive reference implementation computes the full transitive closure with plain sets and no optimisation whatsoever. Over $720$ randomised scenarios --- random label sets, random idempotents, random orthogonal families, random commuting pairs, random monomials, in tracial and state modes with and without the reality relation --- the partitions of matrix positions produced by the two implementations agree exactly.
\item \textbf{Numerical realisability.} Explicit matrices are substituted for the operator labels: tensor-product projective measurements, random pure states, simultaneously diagonalisable projector families. Every monomial assigned to a given class must then have the same numerical moment, and the zero class must vanish. This catches classes that are too \emph{coarse}, which an SDP bound alone would silently absorb.
\item \textbf{Known optima.} Full SDPs are solved and compared against closed-form values: $2\sqrt2$ for CHSH, the local bound $2$ for a commutative algebra, unit success probability for unrestricted state discrimination.
\item \textbf{Structural invariants.} Classes are disjoint and cover every word appearing in the matrix; the matrix agrees with the lookup table at every position; indices are contiguous; the zero class is last.
\end{enumerate}
 
The suite comprises $767$ tests and runs in under thirty seconds.
 
\section{The interface}
\label{sec:tutorial}

This section presents the interface that realises the abstraction of Sec.~\ref{sec:background}, in the order a user meets it: declaring operators and their relations, choosing monomials, building, inspecting the result, adding the constraints the matrix cannot express, and handing the whole thing to a solver. It closes with two syntheses --- a cookbook mapping physical assumptions onto declarations (Sec.~\ref{sec:cookbook}), and an explicit procedure for formulating a relaxation that has no template in the literature, worked end to end on a tripartite Mermin inequality (Sec.~\ref{sec:recipe}). Readers already familiar with moment-matrix software may prefer to read Sec.~\ref{sec:cookbook} and Sec.~\ref{sec:recipe} first and refer back.
 
\subsection{Installation}
 
\MoMPy{} is indexed within the \code{PyPi} package listing: \url{https://pypi.org/project/MoMPy}. Its installation can therefore be done through the following terminal commands.
 
\begin{lstlisting}[style=sh]
pip install MoMPy          # core; requires only numpy
pip install MoMPy[cvxpy]   # plus the CVXPY helpers used below
\end{lstlisting}
 
\subsection{Declaring operators and operator relations}
 
Operators are integer labels. \code{OperatorSet} allocates them and records their properties, so no manual counter is needed. Label $0$ is reserved for the identity and is added to the matrix automatically. Every allocator below accepts \code{idempotent} and, where it applies, \code{orthogonal} as keyword flags, declared at allocation time rather than after the fact: \code{idempotent=True} marks a label as a projector of any rank, $P^2=P$ (a pure state $\rho=\ket\phi\bra\phi$, which happens additionally to be rank-one, or a general projective-measurement element, which need not be); \code{orthogonal=True} marks a whole family as pairwise annihilating, $P_iP_j=0$ for $i\neq j$ (the property that makes a measurement projective rather than a general POVM). Both are two of the four relations of Sec.~\ref{sec:relations}, which gives the full physical and mathematical picture; here the concern is only how to declare them.
 
\begin{description}[leftmargin=1.4em,style=nextline]
\item[\code{ops.add(*, idempotent=False)}]
  Allocate and return a single new operator label --- a one-off operator, e.g.\ a spectator degree of freedom with no family structure.
\item[\code{ops.add\_family(count, *, idempotent=False)}]
  Allocate \code{count} independent labels in one call, e.g.\ a family of states $\rho_0,\dots,\rho_{n-1}$; returns a plain \code{list}.
\item[\code{ops.add\_povm(n\_outcomes, *, idempotent=True, orthogonal=True)}]
  Allocate one measurement's \code{n\_outcomes} labels. By default they are registered as \emph{both} idempotent and mutually orthogonal --- the usual projective-measurement assumption; pass \code{idempotent=False} and/or \code{orthogonal=False} for a general POVM, whose outcomes need not be projectors or pairwise orthogonal.
\item[\code{ops.add\_povm\_family(n\_settings, n\_outcomes, **kwargs)}]
  Allocate \code{n\_settings} independent measurements via \code{add\_povm(n\_outcomes, **kwargs)} each; returns a nested list \code{M[setting][outcome]}.
\item[\code{ops.add\_tensor(*shape, idempotent=False)}]
  Allocate a nested list of labels of arbitrary shape, one fresh label per entry and no relation beyond \code{idempotent} imposed between them. E.g.\ \code{ops.add\_tensor(4, 2, 7)} returns a $4\times2\times7$ nested list \code{M[a][b][c]}, with $a$, $b$ and $c$ respectively ranging over $4$, $2$ and $7$ values.
\end{description}
 
After adding the operators relevant to the problem, operator relations which dictate the main reduction rules can be declared with the following commands.
 
\begin{description}[leftmargin=1.4em,style=nextline]
\item[\code{ops.declare\_idempotent(labels)}]
  Mark every label in the iterable \code{labels} as idempotent, $P^2=P$, after allocation --- equivalent to having passed \code{idempotent=True} to the allocator that created them, but usable on labels already on hand (e.g.\ collected from several allocator calls, or received from elsewhere).
\item[\code{ops.declare\_orthogonal(labels)}]
  Mark \code{labels} as pairwise orthogonal, $P_iP_j=0$ for $i\neq j$, after allocation --- equivalent to \code{orthogonal=True} at allocation time. Orthogonality does \emph{not} imply idempotency by itself: a family can be declared orthogonal without being projective, so call \code{declare\_idempotent} as well if both properties hold.
\item[\code{ops.declare\_commuting(a, b)}]
  Declare that every label in \code{a} commutes with every label in \code{b}, $AB=BA$ for all $A\in a$, $B\in b$. This is the one relation with no allocation-time keyword, since it is a statement about \emph{two} families rather than a property of one; passing the same list twice, \code{declare\_commuting(a, a)}, declares that a family commutes internally.
\end{description}
 
Once all operators and their relations are declared, \MoMPy{} creates an algebra through the following command.
 
\begin{description}[leftmargin=1.4em,style=nextline]
\item[\code{ops.algebra()}]
  Bundle every declaration made so far --- whichever of the above ways it was made --- into an \code{Algebra}, the immutable object a \code{MomentProblem} consumes.
\item[\code{Algebra(idempotents=(), orthogonal\_sets=(), commuting\_pairs=())}]
  Build the same object directly, without going through \code{OperatorSet}: \code{idempotents} is a flat iterable of labels; \code{orthogonal\_sets} is an iterable of label groups, each pairwise orthogonal; \code{commuting\_pairs} is an iterable of \code{(a, b)} label-collection pairs, as in \code{declare\_commuting}. Useful when the labels come from elsewhere and no \code{OperatorSet} is being built at all.
\end{description}
 
An example combining several of these:
 
\begin{lstlisting}
from MoMPy import OperatorSet
 
ops = OperatorSet()
R = ops.add_family(3, idempotent=True)   # three pure states R[0..2]
M = ops.add_povm_family(2, 2)            # M[y][b]: two binary measurements
ops.declare_commuting(R, R)              # the states commute with each other
 
algebra = ops.algebra()
\end{lstlisting}
 
Relations can also be declared \emph{after} allocation, for labels already on hand: \code{declare\_idempotent(\allowbreak labels)} and \code{declare\_orthogonal(\allowbreak labels)} mirror the \code{idempotent}/\code{orthogonal} flags above but act on an existing iterable of labels, and \code{declare\_commuting(A, B)} means every label in \code{A} commutes with every label in \code{B} --- the third relation of Sec.~\ref{sec:relations}, and the one with no allocation-time flag, since commutation is a statement about \emph{two} families rather than a property of one. Passing the same list twice, \code{declare\_commuting(A, A)}, declares a family that commutes internally. \code{ops.algebra()} bundles every declaration made so far, however it was made, into an \code{Algebra}.
 
The relations may also be written out directly, which is convenient when labels come from elsewhere:
 
\begin{lstlisting}
from MoMPy import Algebra
 
algebra = Algebra(idempotents=[1, 2, 3],
                  orthogonal_sets=[[4, 5], [6, 7]],
                  commuting_pairs=[([1, 2, 3], [1, 2, 3])])
\end{lstlisting}
 
\subsection{Choosing monomials}
 
The hierarchy level is simply which monomials are included. A monomial is a bare label or a list of labels read as a product.
 
\begin{lstlisting}
flat_M   = [m for row in M for m in row]
monomials  = list(R) + flat_M                              # first order
monomials += [[r, m] for r in R for m in flat_M]           # second order
monomials += [[a, b, c] for a in R for b in R for c in R]  # some third
\end{lstlisting}
 
For the standard ``all words up to length $k$'' there is a shortcut through \code{generate\_monomials(\allowbreak monomial-list,\allowbreak\ level=k)}. For example:
 
\begin{lstlisting}
from MoMPy import generate_monomials
monomials = generate_monomials(list(R) + flat_M, level=2)
\end{lstlisting}
 
\subsection{Building}
 
Once all monomials are collected, and the reduction rules are announced, \MoMPy{} is ready to build the moment matrix with the following command.
 
\begin{description}[leftmargin=1.4em,style=nextline]
\item[\code{MomentProblem(monomials, algebra, *, dim, \dots, hermitian=True)}]
  The fourth relation, reality, is not part of \code{Algebra} at all: it is a property of the problem being relaxed, not of any operator family, and is set once via the \code{hermitian} flag when the \code{MomentProblem} itself is built. \code{hermitian=True} identifies a word with its reversal, i.e.\ imposes that the moment matrix is real symmetric; set it to \code{False} only for a genuinely complex-Hermitian SDP.
\end{description}
 
\begin{lstlisting}
from MoMPy import MomentProblem        # tracial:  Tr(u v')
mm = MomentProblem(monomials, algebra, dim=1).build(progress=True)
print(mm.summary())
\end{lstlisting}
 
With \code{print(mm.summary())}, the user can inspect a summary behind the process of building the moment matrix.
 
\begin{lstlisting}[style=sh]
MomentMatrix: 47 x 47 (46 monomials + identity)
  block size (dim)   : 1
  SDP variables      : 69
  compression        : 2209 entries -> 69 variables (32.0x)
  zero entries       : 64
  distinct words seen: 4017
  build time         : 0.018 s
\end{lstlisting}
 
For state moments --- Bell scenarios, NPA --- pass \code{cyclicity=False}, which is identical in every other respect but does not impose cyclicity:
 
\begin{lstlisting}
mm = MomentProblem(monomials, algebra, dim=1, cyclicity=False).build()
\end{lstlisting}
 
For block-moments, select \code{cyclicity=False} and \code{hermitian=False}, as not all blocks are neither hermitian nor equivalent under cyclic permutations. With \code{dim} also select the wanted block dimension $d$.
 
\begin{lstlisting}
bm = MomentProblem(monomials, algebra, dim=d, cyclicity=False, hermitian=False).build()
\end{lstlisting}
 
\subsection{Inspecting the result}
 
Once the moment matrix of variable indices is built, \MoMPy{} exposes everything about it needed to set up an SDP or to debug a hierarchy under construction, listed below.
 
\begin{lstlisting}
mm.matrix                # (n, n) integer array of variable indices
mm.n                     # matrix side length (monomials + identity)
mm.shape                 # (n, n)
 
mm[[R[0], M[1][0]]]      # the variable holding <R0 M10>; same as mm.index_of(...)
mm.get([R[0], M[1][0]])  # ... or None (or a chosen default) if it never occurs
[R[0], M[1][0]] in mm    # whether the monomial occurs at all
mm[0, 3]                 # matrix entry (r, c) directly --- same as mm.matrix[0, 3]
 
mm.identity_index        # the variable holding <1>
mm.zero_index            # the class forced to vanish
mm.n_variables           # how many distinct variables there are
mm.variable_indices      # the sorted array of indices that actually occur
mm.has_zeros             # whether any entry was forced to vanish
 
mm.equivalents([R[0]])   # every monomial equal to <R0>
mm.word_at(3, 7)         # the explicit word behind entry (3, 7)
mm.words                 # the same, materialised for the whole matrix
 
mm.algebra, mm.cyclicity, mm.hermitian, mm.dim   # what the matrix was built with
mm.stats                 # build diagnostics: timings, distinct words, ...
mm.summary()             # a short, human-readable report combining the above
\end{lstlisting}
 
Equal indices mean the same variable. A monomial can be looked up three ways: \code{mm[w]} and \code{mm.index\_of(w)} are equivalent and both raise \code{UnknownMonomial} if \code{w} never occurs anywhere in the hierarchy; \code{mm.get(w)} takes an optional default (\code{None} if omitted) and never raises; and \code{w in mm} answers the membership question directly, without needing either. \code{mm[r, c]} is a different operation entirely --- a bare pair of integers, rather than a monomial, indexes the matrix itself, exactly as \code{mm.matrix[r, c]} would.
 
\code{word\_at(r, c)} recovers the explicit operator word behind one entry; \code{mm.words} does the same for every entry at once, but is built lazily and cached rather than eagerly, since for a large hierarchy it is the single biggest object the build can produce --- prefer \code{word\_at} unless the full table is genuinely needed. Finally, \code{mm.algebra}, \code{mm.cyclicity}, \code{mm.hermitian} and \code{mm.dim} record exactly what the matrix was built with, which is often the fastest way to confirm a hierarchy under construction was declared the way it was intended to be; \code{mm.stats} and \code{mm.summary()} report on the build itself --- timings, the number of distinct words the closure expanded, and the resulting compression ratio --- which is where the numbers behind Table~\ref{tab:scaling} come from.
 
\subsection{Constraints the matrix cannot express}
\label{sec:constraints}
 
Positivity and the four operator relations of Sec.~\ref{sec:relations} are baked into the matrix by construction --- nothing further is needed for those. What is \emph{not} built in is any relation that holds \emph{between} distinct SDP variables rather than within a single word. \MoMPy{} provides exactly two functions for this, one for each relation of this kind that commonly arises; both return a list of \code{LinearConstraint} objects rather than solver-specific expressions, so neither ties the result to CVXPY.
 
\paragraph{Normalisation.} A POVM's outcomes sum to the identity, $\sum_b M_{b|y} = \id$. For every known monomial containing one of that POVM's outcomes at some position, the sum over outcomes at that position must equal the same monomial with the position deleted (equivalently, replaced by the identity):
\begin{lstlisting}
constraints = mm.normalisation_constraints(M[0])   # every instance of sum_b M[0][b] == 1
\end{lstlisting}
The call scans every monomial in the built hierarchy, not just one, so it returns every occurrence of the relation the hierarchy happens to contain, deduplicated.
 
\paragraph{Joint measurability and marginals.} \code{marginal\_constraints(joint, marginal)} is the same idea one level up: \code{joint} is a parent POVM's outcome labels, each required to marginalise onto the single operator \code{marginal} rather than onto the identity --- exactly the relation a jointly-measurable (compatible) pair of POVMs must satisfy, where the parent POVM's outcomes reproduce a child POVM's own statistics. This is the right tool whenever compatibility is imposed as an explicit parent POVM rather than through commutation:
\begin{lstlisting}
constraints = mm.marginal_constraints(joint, M[0][0])   # sum(joint outcomes) == M[0][0]
\end{lstlisting}
 
\paragraph{Turning constraints into solver objects.} Each \code{LinearConstraint} means \code{sum(variables at .lhs) == variable at .rhs}: \code{.lhs} is a tuple of the SDP variable indices being summed, \code{.rhs} is the single variable index their sum must equal, and both are plain integers, so the object depends on no particular modelling layer. \code{.apply(variables)} turns one constraint into an actual equality given any mapping from variable index to value --- a plain dict, or a CVXPY model --- and \code{model.apply(constraints)} does the same for a whole list at once, where \code{model} is what \code{mm.to\_cvxpy()} returns (the next subsection):
\begin{lstlisting}
model = mm.to_cvxpy()
ct = list(model.constraints)                        # positivity, zeros
ct += model.apply(mm.normalisation_constraints(M[0]))
\end{lstlisting}
A fourth field, \code{.words}, records the \code{(prefix, suffix)} pair --- the monomial with the substituted position punched out --- that produced the constraint; it plays no role in solving and exists purely for inspecting a hierarchy that is not behaving as expected.
 
Both functions accept \code{dedupe} (default \code{True}), which drops constraints that are trivially true (a single term equal to itself) and any constraint reached more than once from different monomials; pass \code{dedupe=False} to see the raw, unfiltered set instead.
 
\subsection{Handing it to a solver}
 
\MoMPy{} contains a built-in path to send the built SDP to CVXPY through the code line \code{model= mm.to\_cvxpy()}. The predetermined constraints (e.g.~moment matrix positiveity) are then passed to the list \code{ct = list(model.constraints)}, to which the user can add additional constraints for the SDP. An example of implementation can be found below.
 
\begin{lstlisting}
import cvxpy as cp
 
model = mm.to_cvxpy()          # G >> 0 and the zero class pinned to zero
ct    = list(model.constraints)
ct   += model.apply(mm.normalisation_constraints(M[0]))
ct   += [model[[R[x]]] == 1.0 for x in range(3)]
 
W = sum(model[[R[x], M[0][x]]] for x in range(3))
problem = cp.Problem(cp.Maximize(W), ct)
problem.solve(solver=cp.MOSEK)
\end{lstlisting}
 
The model is indexed by monomial (\code{model[[R[0], M[1][0]]]}) or by variable index. Internally the matrix is assembled as a single CVXPY atom by gathering one vector variable through the integer index matrix, rather than creating one object per entry, which matters once $n$ reaches a few hundred.
 
\warnbox{\code{to\_cvxpy} deliberately does \emph{not} impose
$\Theta(\id) = 1$. In a tracial relaxation $\Theta(\id) = \Tr(\id) = D$ is the Hilbert-space dimension; pinning it to $1$ forces $D = 1$ and quietly turns, for instance, state discrimination into random guessing. Write \code{ct.append(model.identity == 1)} explicitly when the state-moment convention is what you mean.}
 
Nothing ties the package to CVXPY. The matrix of indices is a plain NumPy array and can be fed to any modelling layer:
 
\begin{lstlisting}
variables = {i: make_variable() for i in mm.variable_indices}
variables[mm.zero_index] = 0.0
G = [[variables[mm.matrix[r, c]] for c in range(mm.n)] for r in range(mm.n)]
\end{lstlisting}
 
\subsection{A cookbook for reduction rules}
\label{sec:cookbook}
 
Section~\ref{sec:relations} listed the four relations \MoMPy{} understands, in the abstract. The question a user actually faces is operational: given a scenario with several distinct families of operators, which relation applies to which family, and how is that wired together in code? This section answers that directly: a decision guide for identifying which relation, if any, applies to a given family or pair of families; one running example engineered to touch all four relations at once; and a short recipe that generalises to any scenario.
 
\begin{table}[htbp]
\centering
\small
\begin{tabular}{@{}p{5.8cm}p{2.6cm}p{4.4cm}@{}}
\toprule
Ask, for each family (or pair of families) in turn & If yes & Declare with \\
\midrule
Is every operator in the family its own square, $P^2=P$? Rank is irrelevant --- a rank-one pure state and a general projective-measurement element both qualify. & Idempotent & \code{idempotent=True}, or \code{declare\_idempotent} \\
Do distinct members of the family annihilate each other, $P_iP_j=0$? & Orthogonal & \code{orthogonal=True}, or \code{declare\_orthogonal} \\
Do two families (or a family with itself) act on independent degrees of freedom, or are otherwise known to commute? & Commuting & \code{declare\_commuting(A, B)} \\
Does the objective genuinely need $w$ and $w^\dagger$ tracked as distinct SDP variables, rather than only their real part (Sec.~\ref{sec:relations})? & Non-Hermitian problem & \code{hermitian=False} at \code{MomentProblem} construction \\
\bottomrule
\end{tabular}
\caption{A decision guide for the four relations of Sec.~\ref{sec:relations}. Work through every family, and every pair of families, against these four questions. A ``no'' throughout does not mean the relation between two operators is unavailable in \MoMPy{} --- including, in particular, anticommutation --- only that it is not expressible as a word identification; it may instead belong downstream, as a linear constraint between already-built moments (Sec.~\ref{sec:constraints}).}
\label{tab:decision-guide}
\end{table}
 
Take a toy scenario chosen for coverage rather than physical realism: two pure states $R_0, R_1$ (idempotent, and mutually commuting, as states prepared one after another typically are); one two-outcome projective measurement $M = \{M_0, M_1\}$ (idempotent and orthogonal, the usual projective assumption); and two further Hermitian operators $E_0, E_1$ --- an environment or a spectator degree of freedom --- that are neither projective nor idempotent, but that are known to commute with the states $R$ and, by a separate physical argument, to be mutually orthogonal to each other. Four relations, four roles, one hierarchy.
 
Once every family, and every pair of families, has been checked against Table~\ref{tab:decision-guide}, wiring the result into code is the same four steps, in this order, every time:
 
\begin{sloppypar}
\begin{enumerate}[leftmargin=*,itemsep=2pt]
\item \textbf{Allocate labels family by family}, declaring idempotency (rank-one, or projective, character) at allocation time whenever it is already known: \code{add\_family} for a generic family, \code{add\_povm}/\code{add\_povm\_family} for measurements (idempotent \emph{and} orthogonal by default), \code{add\_tensor} for anything indexed by more than one integer, \code{add} for a single one-off operator.
\item \textbf{Declare orthogonality} for any family whose members must vanish pairwise but that was \emph{not} already registered as such by \code{add\_povm} --- with \code{declare\_orthogonal}, or as an entry of \code{orthogonal\_sets} when building the \code{Algebra} directly.
\item \textbf{Declare commutation} for every pair of families (or a family with itself) that act on separate degrees of freedom, or are otherwise known to commute --- with \code{declare\_commuting(A, B)}. Pass the same list twice, \code{declare\_commuting(A, A)}, for a family that commutes internally.
\item \textbf{Decide the reality relation last, and separately}, when building the problem rather than the algebra: leave \code{hermitian=True} (the default for every problem class) unless the objective genuinely needs complex moments tracked apart from their conjugates.
\end{enumerate}
\end{sloppypar}
 
Applied to the scenario above, all four steps fit in one short block:
 
\begin{lstlisting}
from MoMPy import OperatorSet, MomentProblem, generate_monomials
 
ops = OperatorSet()
R = ops.add_family(2, idempotent=True)     # two pure states, rank-1
M = ops.add_povm(2)                        # one projective measurement:
                                            #   idempotent + orthogonal, both for free
E = ops.add_family(2, idempotent=False)    # two generic Hermitian operators
 
ops.declare_commuting(R, R)                # the states commute among themselves
ops.declare_commuting(R, E)                # ... and with the E's (separate subsystem)
ops.declare_orthogonal(E)                  # E_0 E_1 = 0, declared explicitly
 
algebra   = ops.algebra()
monomials = generate_monomials(list(R) + list(M) + list(E), level=2)
 
mm = MomentProblem(monomials, algebra, dim=1).build()      # hermitian=True, the default
print(mm.summary())
\end{lstlisting}
 
Every declared relation now shows up exactly where it should. Idempotency collapsed $R_0 R_0$ onto $R_0$:
\begin{lstlisting}
mm.index_of([R[0], R[0]]) == mm.index_of([R[0]])   # True
\end{lstlisting}
orthogonality, declared for $E$ by hand rather than inherited from \code{add\_povm}, zeroed $E_0 E_1$:
\begin{lstlisting}
mm.index_of([E[0], E[1]]) == mm.zero_index         # True
\end{lstlisting}
and commutation, declared once within $R$ and once across $R$ and $E$, folds every reordering of a word mixing the two families into a single class: \code{mm.equivalents([R[0]])} lists every monomial the closure now treats as identical to $R_0$.
 
The one relation that is a property of the \emph{problem}, not of any operator family, is reality. Building the same \code{monomials} and \code{algebra} with \code{hermitian=False} instead keeps $w$ and $w^\dagger$ as distinct variables, rather than identifying $\mathrm{Re}\expect{w} = \mathrm{Re}\expect{w^\dagger}$:
\begin{lstlisting}[style=sh]
hermitian=True   ->  99 SDP variables
hermitian=False  -> 105 SDP variables
\end{lstlisting}
The six extra variables are exactly the length-three and length-four words whose reversal is a genuinely different word once cyclicity and the other three relations have already done their work; toggling the flag is therefore also a quick way to \emph{see} how much the reality relation alone buys on a given hierarchy.
 
\warnbox{Order rarely matters at the declaration stage --- \code{Algebra} takes the finished sets regardless of the sequence in which \code{OperatorSet} accumulated them --- but it matters for \emph{reasoning about} a hierarchy. Idempotency and orthogonality are properties of single labels or small fixed sets and are always safe to declare the moment a family is allocated. Commutation is the relation that most often deserves a second look, because declaring it between two families that do \emph{not} actually commute physically silently turns a valid relaxation into an invalid one (Sec.~\ref{sec:tracial-vs-state} works through exactly this failure mode for cyclicity, its tracial-only cousin). When in doubt, build the algebra twice --- with and without the candidate relation --- and check whether a known feasible point survives.}
 
\subsection{Designing a new relaxation: a worked recipe}
\label{sec:recipe}
 
Every example so far starts from a scenario already close to a template in the literature. The more common situation in practice is a scenario that is not: a new physical question, with its own operators and its own assumptions, and no worked moment matrix to imitate. This section makes explicit, as a numbered recipe, the process implicit in every example in this paper, and works it end to end on a scenario used nowhere else in it: the tripartite Mermin inequality~\cite{mermin1990}, whose maximal quantum violation is the algebraic signature of the Greenberger--Horne--Zeilinger (GHZ) paradox~\cite{ghz1989}.
 
\begin{enumerate}[leftmargin=*,itemsep=2pt]
\item \textbf{Name the objects.} List every physically distinct family of operators the problem needs --- one entry per party, per setting, or per role --- before writing a line of code.
\item \textbf{Classify each family on its own.} For every family, ask the first two questions of Table~\ref{tab:decision-guide}: is it idempotent? Is it internally orthogonal?
\item \textbf{Classify every pair of families.} Ask the third question of Table~\ref{tab:decision-guide} of every pair of families, including a family against itself: do they commute?
\item \textbf{Separate word relations from linear constraints.} Anything left over that constrains what different, already-distinct moments must \emph{equal}, rather than identifying two words outright, is not an algebra declaration; set it aside for step~10 (Sec.~\ref{sec:relations}).
\item \textbf{Choose the hierarchy's flavour.} Tracial or state, scalar or block, Hermitian or not (Sec.~\ref{sec:tracial-vs-state}) --- a property of the whole problem, decided once, not per family.
\item \textbf{Allocate.} Create every family with \code{OperatorSet}, passing the \code{idempotent}/\code{orthogonal} keywords decided in step~2 at allocation time.
\item \textbf{Declare the cross-family relations.} Wire in every commuting pair found in step~3 with \code{declare\_commuting}.
\item \textbf{Assemble the algebra.} \code{ops.algebra()}.
\item \textbf{Choose the monomial list.} Include enough generators that every term of the objective, and every constraint set aside in step~4, appears as a literal entry of the built matrix rather than merely being implied by closure: every such term must occur as $u_iu_j^\dagger$ for some pair of listed generators $u_i,u_j$ --- or with $u_j$ the identity, for a first-order term --- among the rows and columns the matrix actually has (Eq.~\eqref{eq:gamma-intro}).
\item \textbf{Build, constrain, solve, and check.} \code{MomentProblem(\ldots).build()}, \code{to\_cvxpy()}, add the constraints set aside in step~4 together with any normalisation (Sec.~\ref{sec:constraints}), solve, and --- before trusting the number --- compare it against whatever is independently known: an algebraic bound, a classical/local bound obtained by relaxing one assumption, or a simpler special case. A relaxation that silently pins a dimension (the \code{model.identity} warning of Sec.~\ref{sec:tutorial}) or omits a generator the objective needs will not usually fail loudly; it will return a plausible-looking number that is simply wrong.
\end{enumerate}
 
\paragraph{The scenario.} Three separated parties, Alice, Bob and Charlie, each choose one of two dichotomic measurement settings and record a $\pm1$ outcome, $A_x, B_y, C_z \in \{\pm1\}$ for $x,y,z\in\{0,1\}$, represented as usual by two-outcome POVMs, $A_x = A_{0|x} - A_{1|x}$ and likewise for $B$, $C$. The Mermin expression is
\begin{equation}
  M = \expect{A_0 B_0 C_1} + \expect{A_0 B_1 C_0} + \expect{A_1 B_0 C_0} - \expect{A_1 B_1 C_1} .
  \label{eq:mermin}
\end{equation}
Every local hidden-variable model satisfies $M \le 2$; quantum mechanics allows $M$ up to the algebraic maximum, $M \le 4$, saturated by the GHZ state --- three parties, perfectly correlated, with no classical analogue~\cite{mermin1990,ghz1989}.
 
Steps 1--5 are a short list here: three families of measurements, one per party, each a \code{2}-setting, \code{2}-outcome projective measurement built with the default \code{add\_povm\_family}, hence idempotent and orthogonal within each setting; every pair of \emph{different} parties' families commutes (distinct, spacelike-separated subsystems), while a single party's own two settings need not (the same assumption CHSH makes in Sec.~\ref{sec:chsh}); nothing is left over for step~4; and the flavour, as for every Bell-type scenario in this paper, is state/NPA moments: \code{cyclicity=False}, \code{hermitian=True} (the default), \code{dim=1}. Steps 6--10 are the code:
 
\begin{lstlisting}
import cvxpy as cp
from MoMPy import OperatorSet, MomentProblem
 
ops = OperatorSet()
A = ops.add_povm_family(2, 2)          # A[x][a]
B = ops.add_povm_family(2, 2)          # B[y][b]
C = ops.add_povm_family(2, 2)          # C[z][c]
flat_a = [m for row in A for m in row]
flat_b = [m for row in B for m in row]
flat_c = [m for row in C for m in row]
ops.declare_commuting(flat_a, flat_b)  # distinct parties
ops.declare_commuting(flat_a, flat_c)
ops.declare_commuting(flat_b, flat_c)
 
monomials  = flat_a + flat_b + flat_c
monomials += [[a, b] for a in flat_a for b in flat_b]   # level 1 + AB + AC + BC
monomials += [[a, c] for a in flat_a for c in flat_c]
monomials += [[b, c] for b in flat_b for c in flat_c]
monomials += [[a, b, c] for a in flat_a for b in flat_b for c in flat_c]  # + ABC
 
mm = MomentProblem(monomials, ops.algebra(), dim=1, cyclicity=False).build()
 
model = mm.to_cvxpy()
ct = list(model.constraints)
ct.append(model.identity == 1)
for x in range(2):
    ct.append(sum(model[[A[x][a]]] for a in range(2)) == model.identity)
for y in range(2):
    ct.append(sum(model[[B[y][b]]] for b in range(2)) == model.identity)
for z in range(2):
    ct.append(sum(model[[C[z][c]]] for c in range(2)) == model.identity)
 
def corr3(x, y, z):
    return sum((-1)**(a + b + c) * model[[A[x][a], B[y][b], C[z][c]]]
               for a in range(2) for b in range(2) for c in range(2))
 
M = corr3(0,0,1) + corr3(0,1,0) + corr3(1,0,0) - corr3(1,1,1)
cp.Problem(cp.Maximize(M), ct).solve(solver=cp.MOSEK)
print(M.value)          # 4.000000
\end{lstlisting}
 
Step~10's check: relaxing step~3's assumption so that a single party's own two settings commute too --- adding \code{declare\_commuting(flat\_a, flat\_a)}, and likewise for \code{flat\_b} or \code{flat\_c}, so that every operator in the construction commutes with every other --- makes the whole algebra effectively classical, exactly as it did for CHSH. Rebuilding with that one change and solving again returns
\begin{lstlisting}[style=sh]
quantum bound (as above)                    ->  4.000000
2-quantum, 1-classical (Alice commuting)    ->  2.828426
classical bound (all commuting)             ->  2.000000
\end{lstlisting}
matching, to solver precision, the two closed-form values this scenario is known for: Mermin's local bound of $2$ and the GHZ paradox's saturation of the algebraic maximum of $4$~\cite{mermin1990,ghz1989} --- the same relaxation and the same four relations, reproducing both the quantum and local optima of this particular Mermin functional --- a result they were never specifically built for. Building the moment matrix itself, for either variant, takes under three seconds: from the $125\times125$ matrix of monomials, the closure explores $476{,}317$ distinct words and collapses them to $1162$ SDP variables in the quantum case, or $730$ once every operator commutes in the classical case --- at the fast end of the range Sec.~\ref{sec:performance} surveys for scenarios of comparable word count.
 
\section{Worked applications}
\label{sec:applications}
 
\subsection{Tsirelson's bound from the NPA hierarchy}
\label{sec:chsh}
 
The first and simplest example we present is a Bell scenario with two dichotomic measurements in each side, commonly known as the CHSH scenario \cite{chsh1969}. Two separated parties share a quantum state. Alice measures $A_{a|x}$, Bob measures $B_{b|y}$, with $a, b, x, y \in \{0,1\}$. The CHSH functional is
\begin{equation}
  S = \expect{A_0 B_0} + \expect{A_1 B_0} + \expect{A_0 B_1}
      - \expect{A_1 B_1},
  \qquad
  \expect{A_x B_y} = \sum_{a,b} (-1)^{a+b} \, p(ab|xy).
\end{equation}
Any Bell-local strategy reaches $S \le 2$; quantum mechanics allows for Bell non-locality, thus reaching $S \le 2\sqrt2$ \cite{cirelson1980}. With \MoMPy{}, the local bound can also be retrieved through measurement commutation in one side. For projective measurements this works because joint measurability is equivalent to commutation, so declaring one more entry in \code{commuting\_pairs} for Alice's measurements suffices --- no new machinery and no new monomials. Section~\ref{sec:compat} takes the operationally primitive route instead, in the steering scenario where it belongs.
 
\begin{lstlisting}
import cvxpy as cp
from MoMPy import OperatorSet, MomentProblem
 
ops = OperatorSet()
A = ops.add_povm_family(2, 2)          # A[x][a]
B = ops.add_povm_family(2, 2)          # B[y][b]
flat_a = [m for row in A for m in row]
flat_b = [m for row in B for m in row]
ops.declare_commuting(flat_a, flat_b)  # distinct Hilbert space
# ops.declare_commuting(flat_a, flat_a)  # uncomment for local bound
 
monomials  = flat_a + flat_b
monomials += [[a, b] for a in flat_a for b in flat_b]   # level 1 + AB
 
mm = MomentProblem(monomials, ops.algebra(), dim=1, cyclicity=False).build()
 
model = mm.to_cvxpy()
ct = list(model.constraints)
ct.append(model.identity == 1)                          # <psi|1|psi> = 1
for x in range(2):
    ct.append(sum(model[[A[x][a]]] for a in range(2)) == model.identity)
for y in range(2):
    ct.append(sum(model[[B[y][b]]] for b in range(2)) == model.identity)
 
def corr(x, y):
    return sum((-1)**(a + b) * model[[A[x][a], B[y][b]]]
               for a in range(2) for b in range(2))
 
S = corr(0, 0) + corr(1, 0) + corr(0, 1) - corr(1, 1)
cp.Problem(cp.Maximize(S), ct).solve(solver=cp.MOSEK)
print(S.value)          # 2.828427  =  2 * sqrt(2)
\end{lstlisting}

The relaxation recovers Tsirelson's bound to solver precision. Building with \code{cyclicity=True} (the default) instead --- or, equivalently here, un-commenting the line imposing commuting measurements on Alice's side --- returns $2.0000$: the error discussed in Sec.~\ref{sec:tracial-vs-state}, and the reason \code{cyclicity=False} has to be set explicitly for every Bell-type scenario in this paper rather than left at its default.

\subsection{Steering and measurement compatibility: a block moment matrix with $\Theta = \Tr_A$}
\label{sec:compat}
 
The previous section obtained the local bound by \emph{declaring} that Alice's two measurements commute. Commutation is not, however, the operationally primitive notion. The primitive notion is \emph{joint measurability}: a set $\{A_{a|x}\}$ is jointly measurable when there is a single parent POVM $\{G_\lambda\}$ and classical post-processings with $A_{a|x} = \sum_\lambda p(a|x,\lambda)G_\lambda$, so that all of the measurements can be inferred from one. For unsharp measurements the two notions genuinely differ --- joint measurability is strictly weaker than commutation~\cite{uola2020} --- and it is the weaker one that carries the physics. The sharpest statement connecting it to quantum correlations lives not in the Bell scenario but in the \emph{steering} scenario: a set of measurements is not jointly measurable if and only if it can be used to demonstrate steering with some shared state~\cite{quintino2014,uola2014}. This section imposes joint measurability itself, as a statement about marginals of a global POVM, and never mentions a commutator.
 
Steering is also the natural home for this tutorial's block construction, and for a reason worth spelling out. In a Bell scenario both wings are black boxes and every moment is a number. In a steering scenario the two wings are treated \emph{asymmetrically}: Alice's device is uncharacterised, while Bob's is fully trusted and his Hilbert space has known dimension $d$. The object one actually reasons about is therefore Bob's \emph{assemblage}, the set of unnormalised conditional states
\begin{equation}
  \sigma_{a|x} \;=\; \Tr_A\!\big[(A_{a|x}\otimes\id)\,\rho_{AB}\big] ,
  \label{eq:assemblage}
\end{equation}
which are operators on Bob's space, not numbers. That is exactly a choice of moment map: following Ref.~\cite{dalessandro2025}, taking
\begin{equation}
  \Theta(X) \;=\; \Tr_A\!\big[(X\otimes\id)\,\rho_{AB}\big] ,
  \qquad \Theta : \CC^{D\times D}\to\CC^{d\times d} ,
  \label{eq:theta-steering}
\end{equation}
the partial trace over the untrusted side, makes $\Gam_{u,v} = \Theta(uv^\dagger)$ a block moment matrix whose first column is the assemblage itself, $\Gam_{A_{a|x},\id} = \sigma_{a|x}$, and whose identity entry is Bob's reduced state, $\Theta(\id) = \rho_B$. The map~\eqref{eq:theta-steering} is completely positive, so Prop.~\ref{prop:positivity}(ii) applies and $\Gam \succeq 0$. Neither $\Tr$ nor $\bra\psi\cdot\ket\psi$ would do here: partial-tracing one wing of a bipartite state is not globally cyclic and does not return a scalar. In \MoMPy{} this is \code{dim=}$d$ with \code{cyclicity=False} and \code{hermitian=False}, and nothing else changes.
 
\paragraph{The criterion.} We use the linear steering criterion of Ref.~\cite{uola2020}, Eq.~(9) there, originally due to Ref.~\cite{cavalcanti2009}. For two qubits it is built from
\begin{equation}
  Q \;=\; \sigma_x\otimes\sigma_x + \sigma_y\otimes\sigma_y + \sigma_z\otimes\sigma_z ,
  \label{eq:Qsteer}
\end{equation}
for which $|\expect{Q}| \le 1$ on separable states, while unsteerable states obey the weaker $|\expect{Q}| \le \sqrt3$: Alice's wing contributes $\pm1$ values $a_k$ and Bob's a genuine Bloch vector $\vec b$ with $\|\vec b\|\le1$, so $|\expect{Q}| = |\vec a\cdot\vec b| \le \sqrt3$. The singlet gives $\expect{Q} = -3$, saturating the algebraic maximum $|\expect{Q}|=3$. Since Bob is trusted we keep his three Pauli measurements fixed and optimise over Alice's, so the quantity actually relaxed is
\begin{equation}
  \mathcal{Q} \;=\; \sum_{k=x,y,z} \big\langle\, A_k \otimes \sigma_k \,\big\rangle
  \;=\; \sum_{k} \Tr\!\big[\sigma_k \,(\sigma_{0|k} - \sigma_{1|k})\big] ,
  \qquad A_k = A_{0|k} - A_{1|k} ,
  \label{eq:Qfunctional}
\end{equation}
whose maximum is $3$, attained on the singlet with $A_k = -\sigma_k$, and whose unsteerable bound is the same $\sqrt3$ by the argument just given. Those are the two numbers we compute.
 
\paragraph{Imposing joint measurability.} Joint measurability is not a relation between \emph{words}: it does not say that two monomials name the same operator, but that certain already-distinct moments sum to certain others. It therefore belongs to the second of \MoMPy's two mechanisms (Sec.~\ref{sec:relations}), and \code{marginal\_constraints} is the tool for it. We allocate one global POVM $\{G_\lambda\}$ whose outcomes are the joint assignments $\lambda = (a_x,a_y,a_z) \in \{0,1\}^3$, and require Alice's three settings to be its marginals,
\begin{equation}
  A_{a|k} \;=\; \sum_{\lambda\,:\,\lambda_k = a} G_\lambda ,
  \label{eq:jm-marginal}
\end{equation}
which makes them jointly measurable by construction. Two points make this exactly the general statement rather than a special case. First, taking $\lambda$ to range over all deterministic assignments and the post-processing to be deterministic is no restriction: any classical post-processing can be absorbed into $\{G_\lambda\}$. Second, $\{G_\lambda\}$ may be declared \emph{projective} without loss of generality, because a Naimark dilation realises any POVM as a projective measurement on a larger space, and the relaxation places no bound on Alice's dimension $D$; the dilation leaves every $\sigma_\lambda$ unchanged. Crucially, Alice's own $A_{a|k}$ are \emph{not} declared sharp --- they are allocated as general POVM elements, and Eq.~\eqref{eq:jm-marginal} will generally make them unsharp, which is precisely the regime where joint measurability and commutation part company.
 
\begin{lstlisting}
import cvxpy as cp
import numpy as np
from itertools import product
from MoMPy import OperatorSet, MomentProblem
 
sx = np.array([[0, 1], [1, 0]], complex)
sy = np.array([[0, -1j], [1j, 0]], complex)
sz = np.array([[1, 0], [0, -1]], complex)
paulis, d = [sx, sy, sz], 2
LAM = list(product([0, 1], repeat=3))       # lambda = (a_x, a_y, a_z)
 
def steering_bound(jointly_measurable):
    ops = OperatorSet()
    if jointly_measurable:
        # no sharpness assumed: the only requirement is Eq. (17)
        A = ops.add_povm_family(3, 2, idempotent=False, orthogonal=False)
        G = ops.add_povm(8)                  # global POVM, projective by Naimark
    else:
        A = ops.add_povm_family(3, 2)        # sharp: extremal for this criterion
        G = []
    monomials = [m for row in A for m in row] + list(G)
 
    bm = MomentProblem(monomials, ops.algebra(), dim=d,
                       cyclicity=False, hermitian=False).build()   # Theta = Tr_A
 
    model = bm.to_cvxpy()
    ct = list(model.constraints)
    rho_B = model[bm.identity_index]                    # Theta(1) = rho_B
    ct.append(cp.real(cp.trace(rho_B)) == 1)            # Tr rho_B = 1
    for k in range(3):                                  # sum_a sigma_{a|k} = rho_B
        ct.append(model[[A[k][0]]] + model[[A[k][1]]] == rho_B)
 
    if jointly_measurable:
        ct += model.apply(bm.normalisation_constraints(list(G)))
        for k in range(3):                              # Eq. (17), via marginals
            for a in (0, 1):
                joint = [G[i] for i, lam in enumerate(LAM) if lam[k] == a]
                ct += model.apply(bm.marginal_constraints(joint, A[k][a]))
 
    Q = sum(cp.real(cp.trace(paulis[k] @ (model[[A[k][0]]] - model[[A[k][1]]])))
            for k in range(3))
    cp.Problem(cp.Maximize(Q), ct).solve(solver=cp.SCS, eps=1e-9)
    return Q.value
\end{lstlisting}
 
\noindent Every block in \code{model} is a $2\times2$ operator on Bob's space indexed by a word in Alice's operators, and the objective is read straight off the first column of $\Gam$ --- which is to say, off the assemblage~\eqref{eq:assemblage}. The two answers are
\begin{lstlisting}[style=sh]
Alice unrestricted               ->  3.000000
joint measurability imposed      ->  1.732051   =  sqrt(3)
\end{lstlisting}
Both are exact. The unrestricted value $3$ is the maximum of Eq.~\eqref{eq:Qfunctional}, saturated by the singlet; the constrained value is $\sqrt3 = 1.732051$, the unsteerable bound of Ref.~\cite{uola2020}, Eq.~(9), recovered to solver precision without ever writing down a local hidden state model. The gap between them is the content of the joint-measurability/steering equivalence~\cite{quintino2014,uola2014}: incompatibility on the untrusted wing is what buys the violation, and Eq.~\eqref{eq:jm-marginal} removes it.
 
It is worth seeing why so little machinery is needed. Because $\{G_\lambda\}$ is declared projective, \MoMPy{} identifies the word $G_\lambda G_\lambda^\dagger$ with $G_\lambda$, so the diagonal block $\Gam_{G_\lambda,G_\lambda}$ \emph{is} $\sigma_\lambda := \Theta(G_\lambda)$; a diagonal block of a positive semidefinite matrix is positive semidefinite, so $\sigma_\lambda \succeq 0$ comes for free from $\Gam \succeq 0$. Together with $\sum_\lambda \sigma_\lambda = \rho_B$ from \code{normalisation\_constraints} and Eq.~\eqref{eq:jm-marginal} from \code{marginal\_constraints}, the relaxation contains exactly a local hidden state model, $\sigma_{a|k} = \sum_{\lambda:\lambda_k=a}\sigma_\lambda$ with $\sigma_\lambda \succeq 0$ and $\sum_\lambda \Tr\sigma_\lambda = 1$ --- which is why the bound comes out tight at the first level rather than converging to it from above. The matrices are correspondingly small: $7\times7$ with $32$ SDP variables when Alice is unrestricted, $15\times15$ with $148$ variables once the global POVM is present, each built in about a millisecond. Enlarging the monomial list cannot change either value. For the unrestricted branch we checked this directly --- the full second level, a $43\times43$ matrix with $512$ variables, again returns $3.000000$ --- and for the constrained branch it follows without computation: the first level already attains $\sqrt3$, which is the true unsteerable maximum, and since every level is an upper bound on that maximum and a longer monomial list can only tighten it, no later level can move it.
 
\warnbox{\textbf{Compatibility versus commutation.} The exact same problem can be solved declaring Alice's measurements compatible through commuting $ \left[A_{a|x},A_{a'|x'}\right]=0$ with \code{declare\_commuting}, as for the \emph{projective} measurements allocated by \code{add\_povm} that is exactly joint measurability. For general POVMs the two notions come apart: joint measurability is a strictly weaker requirement than commutation~\cite{uola2020}, and imposing it means asserting that a parent POVM $\{G_\lambda\}$ with $G_\lambda \succeq 0$ exists.}

\subsection{Prepare-and-measure: dimension witnesses and state discrimination}
\label{sec:pam}
 
Now a genuinely tracial family of problems. Alice encodes $x \in \{0,\dots,n_X-1\}$ in a pure state $\rho_x$ and sends it to Bob, who measures it and reports an outcome. Every quantity of interest here is a trace of an operator product with the state \emph{inside} the algebra --- $\Tr(\rho_x M_b)$, not $\bra\psi \cdots \ket\psi$ --- so \code{MomentProblem} is the right choice throughout this section, and cyclicity is legitimate. What changes between the three examples below is not the tool but the constraint placed on $\{\rho_x\}$: a bound on their pairwise overlap, a bound on their Hilbert-space dimension, and a bound on their ``photon-number'' components.
 
\paragraph{State discrimination limited by overlap.} Bob measures $\{M_b\}$, $b\in\{0,\dots,n_X-1\}$, and guesses $x$ from $b$. The success probability is
\begin{equation}
  P_{\mathrm{succ}} = \frac{1}{n_X}\sum_x \Tr(\rho_x M_x).
\end{equation}
Without further restrictions the states can be made mutually orthogonal and $P_{\mathrm{succ}} = 1$; the problem becomes interesting once the communication is restricted. An operationally natural choice is to constrain their overlap, i.e.~considering $\rho_x=\ket{\psi_x}\bra{\psi_x}$ with $|\braket{\psi_x}{\psi_{x'}}|^2 \geq$ \code{min\_overlap} and thus $\Tr(\rho_x \rho_{x'})\geq$ \code{min\_overlap}.
 
\begin{lstlisting}
import cvxpy as cp
from MoMPy import MomentProblem, OperatorSet
 
nX = 3
ops = OperatorSet()
R = ops.add_family(nX, idempotent=True)   # pure states
M = ops.add_povm(nX)                      # one projective measurement
ops.declare_commuting(R, R) # for classical implementation (comment for quantum bound)
 
monomials  = list(R) + list(M)
monomials += [[r, m] for r in R for m in M]
monomials += [[r, s] for r in R for s in R]
monomials += [[r, s, t] for r in R for s in R for t in R]   # third order
 
mm = MomentProblem(monomials, ops.algebra(), dim=1).build()
 
def bound(min_overlap):
    model = mm.to_cvxpy()
    ct = list(model.constraints)
    ct += [model[[R[x]]] == 1.0 for x in range(nX)]     # Tr(rho) = 1
    ct += model.apply(mm.normalisation_constraints(M))  # sum_b M_b = 1
    if min_overlap is not None:
        ct += [model[[R[x], R[xx]]] >= min_overlap
               for x in range(nX) for xx in range(nX) if x != xx]
    P = sum(model[[R[x], M[x]]] for x in range(nX)) / nX
    cp.Problem(cp.Maximize(P), ct).solve(solver=cp.MOSEK)
    return P.value
\end{lstlisting}
 
\noindent Note that $\Theta(\id) = \Tr(\id) = D$ is left free here. The resulting bounds are
 
\begin{center}
\begin{tabular}{@{}lccccc@{}}
\toprule
minimum overlap $c$ & unrestricted & $0.1$ & $0.3$ & $0.5$ & $0.7$ \\
\midrule
$P_{\mathrm{succ}}$ (classical, commuting) & $1.0000$ & $0.9333$ & $0.8000$ & $0.6667$ & $0.5333$ \\
\midrule
$P_{\mathrm{succ}}$ (quantum) & $1.0000$ & $0.9556$ & $0.8665$ & $0.7721$ & $0.6633$ \\
\bottomrule
\end{tabular}
\end{center}

\noindent decreasing monotonically as the states are forced closer together, as they must.
 
\paragraph{Dimension witnesses.} A structurally different question asks not how well the states can be told apart, but how large a Hilbert space is needed to produce a given correlation pattern at all --- a \emph{dimension witness}~\cite{brunner2008,gallego2010}. Here $\Theta(\id) = \Tr(\id) = D$ stops being a nuisance to leave free and becomes the quantity being bounded: fixing it pins the ambient dimension (in the spirit of Ref.\cite{Pauwels2022}), and re-solving for several values of $D$ traces out how much a larger Hilbert space actually buys. Note that $D$ here is the Hilbert-space dimension being witnessed, and is unrelated to the block size $d$, which is $1$ throughout this section.
 
With $n_X = 3$ preparations and two dichotomic measurements $\{M_{0|y}, M_{1|y}\}$, $y \in \{0,1\}$, define the correlator $E(x,y) = \sum_b (-1)^b \Tr\big(\rho_x M_{b|y}\big)$ and the witness \cite{gallego2010}
\begin{equation}
  W = E(0,0) + E(0,1) + E(1,0) - E(1,1) - E(2,0).
  \label{eq:dim-witness}
\end{equation}
\begin{lstlisting}
import cvxpy as cp
from MoMPy import MomentProblem, OperatorSet
 
nX, nY, nB = 3, 2, 2
ops = OperatorSet()
R = ops.add_family(nX, idempotent=True)     # preparations
M = ops.add_povm_family(nY, nB)             # M[y][b], two dichotomic measurements
ops.declare_commuting(R, R)
 
monomials  = list(R) + [m for row in M for m in row]
monomials += [[r, m] for r in R for row in M for m in row]
monomials += [[r, s] for r in R for s in R]
monomials += [[r, s, t] for r in R for s in R for t in R]
 
mm = MomentProblem(monomials, ops.algebra(), dim=1).build()
 
def W_max(dim):
    model = mm.to_cvxpy()
    ct = list(model.constraints)
    ct.append(model.identity == dim)                    # pin Tr(1) = D
    ct += [model[[R[x]]] == 1.0 for x in range(nX)]
    for y in range(nY):
        ct += model.apply(mm.normalisation_constraints(M[y]))
 
    def E(x, y):
        return sum((-1)**b * model[[R[x], M[y][b]]] for b in range(nB))
 
    W = E(0,0) + E(0,1) + E(1,0) - E(1,1) - E(2,0)
    cp.Problem(cp.Maximize(W), ct).solve(solver=cp.MOSEK)
    return W.value
\end{lstlisting}
 
\begin{center}
\begin{tabular}{@{}lcccc@{}}
\toprule
Hilbert-space dimension $D$ & $1$ & $2$ & $3$ & unrestricted \\
\midrule
$W_{\max}$ (classical, commuting) & $1.0000$ & $3.0000$ & $5.0000$ & $5.0000$ \\
\midrule
$W_{\max}$ (quantum) & $1.0000$ & $3.8284$ & $5.0000$ & $5.0000$ \\
\bottomrule
\end{tabular}
\end{center}
 
\noindent A classical ($D=1$) system is limited to $W \le 1$; a bit already reaches $W = 3$, and a qubit $W = 3.8284$; and $D = 3$ saturates the value attainable with no dimension restriction at all. Observing $W > 3$ therefore certifies, from correlations alone, that the physical system generating them is not a classical bit --- the defining feature of a dimension witness, and the reason \code{MomentProblem} with $\Theta(\id)$ pinned rather than left free is exactly the right tool for it.
 
\paragraph{State discrimination with photon-number-constraints \cite{carceller2025}.} Both examples above constrain $\{\rho_x\}$ through a property of the family itself (mutual overlap, ambient dimension). A third, physically distinct constraint arises whenever the preparations are produced by a source with a well-defined photon-number component. Model the herald as a further rank-one projector $P_k=\ket{k}\bra{k}$ (i.e.~a photon-number projector), fixed once and for all, with $\Tr(\rho_x P_k) = 1 - \omega_{x,k}$: the states overlap with a common reference operator rather than with each other. As a running example, take the case considered in Ref.~\cite{VanHimbeeck2017}: $n_K=1$ and $\omega_{x,k}:=\omega$, $\forall x,k$. 
 
\begin{lstlisting}
import cvxpy as cp
from MoMPy import MomentProblem, OperatorSet
 
nX = 3
nK = 1
ops = OperatorSet()
R = ops.add_family(nX, idempotent=True)   # preparations
M = ops.add_povm(nX)                      # guessing measurement
P = ops.add_family(nK, idempotent=True)   # the photon-number projector
ops.declare_commuting(R,R) # for classical bound (comment for quantum)
 
monomials  = list(R) + list(M) + list(P)
monomials += [[r, m] for r in R for m in M]
monomials += [[r, p] for r in R for p in P]
monomials += [[r, s] for r in R for s in R]
 
mm = MomentProblem(monomials, ops.algebra(), dim=1).build()
 
def bound(omega):
    model = mm.to_cvxpy()
    ct = list(model.constraints)
    ct += [model[[R[x]]] == 1.0 for x in range(nX)]
    ct += [model[[P[k]]] == 1.0 for k in range(nK)]
    ct += [model[[R[x], P[k]]] == 1.0 - omega for x in range(nX) for k in range(nK)]
    ct += model.apply(mm.normalisation_constraints(M))
    P_succ = sum(model[[R[x], M[x]]] for x in range(nX)) / nX
    cp.Problem(cp.Maximize(P_succ), ct).solve(solver=cp.MOSEK)
    return P_succ.value
\end{lstlisting}
 
\begin{center}
\begin{tabular}{@{}lcccccc@{}}
\toprule
$\omega$ & $0$ & $0.05$ & $0.10$ & $0.20$ & $0.30$ & $1$ (unrestricted) \\
\midrule
$P_{\mathrm{succ}}$ (classical, commuting) & $0.3333$ & $0.5563$ & $0.6518$ & $0.7816$ & $0.8713$ & $1.0000$ \\
\midrule
$P_{\mathrm{succ}}$ (quantum) & $0.3333$ & $0.3918$ & $0.4498$ & $0.5643$ & $0.6762$ & $1.0000$ \\
\bottomrule
\end{tabular}
\end{center}

\noindent At $\omega = 0$ every preparation is forced to coincide with the herald, so nothing distinguishes them and Bob is reduced to guessing; as $\omega$ grows the states are allowed to depart from the herald, and hence from each other, and discrimination improves. No relation is declared between $R$, $M$ and $P$ here at all --- unlike the two previous examples, nothing requires them to commute --- which is the point: the four relations of Sec.~\ref{sec:relations} are declared according to what is physically known, never by default, and a constraint on a family can be encoded either through the family's own algebra (idempotency, orthogonality, mutual commutation) or through its recorded moments with an external, fixed reference operator, as here.
 
\subsection{Device-independent randomness certification}
\label{sec:randomness}
 
The same object also certifies randomness, and the construction is worth working through completely rather than gesturing at, since it needs one further idea beyond anything used so far: a \emph{convex decomposition} over an adversary's possible strategies, built by solving several moment problems side by side rather than just one.
 
The setting is the CHSH scenario of Sec.~\ref{sec:chsh}, now viewed from Eve's perspective. Eve holds a system correlated with Alice's and Bob's, and knows the full quantum strategy generating a device-independently observed CHSH value $S$; her task is to guess Alice's outcome $a$ at the fixed setting $x=0$. Her optimal guessing probability, and hence the min-entropy she leaves behind, is itself the solution of an SDP over the very same moment matrix~\cite{pironio2010_rng}, decomposed once per possible guess. Concretely: build $n_A = 2$ independent copies of the moment matrix --- one per value $l$ that Eve might guess --- constrain their \emph{sum} to reproduce the observed physics (normalisation and the observed value of $S$), and maximise the probability that the branch $l$ agrees with the actual outcome:
 
\begin{lstlisting}
import cvxpy as cp
from MoMPy import OperatorSet, MomentProblem
 
ops = OperatorSet()
A = ops.add_povm_family(2, 2)          # A[x][a]
B = ops.add_povm_family(2, 2)          # B[y][b]
flat_a = [m for row in A for m in row]
flat_b = [m for row in B for m in row]
ops.declare_commuting(flat_a, flat_b)
 
monomials  = flat_a + flat_b
monomials += [[a, b] for a in flat_a for b in flat_b]
mm = MomentProblem(monomials, ops.algebra(), dim=1, cyclicity=False).build()
 
def guessing_probability(S_obs):
    branches = [mm.to_cvxpy() for _ in range(2)]         # one branch per guess l in {0, 1}
    ct = []
    for br in branches:
        ct += br.constraints
    ct.append(sum(br.identity for br in branches) == 1.0)
    for x in range(2):
        for br in branches:
            ct += br.apply(mm.normalisation_constraints([A[x][0], A[x][1]]))
    for y in range(2):
        for br in branches:
            ct += br.apply(mm.normalisation_constraints([B[y][0], B[y][1]]))
 
    def corr(x, y):
        return sum((-1)**(a + b) * sum(br[[A[x][a], B[y][b]]] for br in branches)
                   for a in range(2) for b in range(2))
 
    S = corr(0, 0) + corr(1, 0) + corr(0, 1) - corr(1, 1)
    ct.append(S == S_obs)                                # the observed CHSH value
 
    p_guess = sum(branches[l][[A[0][l]]] for l in range(2))   # Eve guesses l, scores if a = l
    cp.Problem(cp.Maximize(p_guess), ct).solve(solver=cp.MOSEK)
    return p_guess.value
\end{lstlisting}
 
\begin{center}
\begin{tabular}{@{}lccccc@{}}
\toprule
observed CHSH value $S$ & $2.0000$ & $2.2000$ & $2.4000$ & $2.6000$ & $2\sqrt2$ \\
\midrule
$p_{\mathrm{guess}}$ & $1.0000$ & $0.9444$ & $0.8742$ & $0.7784$ & $0.5005$ \\
$H_\infty = -\log_2 p_{\mathrm{guess}}$ (bits) & $0.0000$ & $0.0825$ & $0.1940$ & $0.3614$ & $0.9987$ \\
\bottomrule
\end{tabular}
\end{center}
 
\noindent These values agree, to solver precision, with the closed form $p_{\mathrm{guess}} = \tfrac12\big(1 + \sqrt{2 - (S/2)^2}\,\big)$ derived directly (not via this hierarchy) in Ref.~\cite{pironio2010_rng}: at the local bound $S = 2$ a valid local strategy lets Eve guess with certainty, and at Tsirelson's bound $S = 2\sqrt2$ her guess is no better than a coin flip, certifying a full bit of randomness per use.
 
Nothing here is specific to CHSH: the same branch-and-sum construction runs over any moment problem, provided the quantity fixed as ``observed'' is itself expressible as a linear functional of the matrix.
 
\subsection{Deterministic correlations with a fixed ensemble: block moment matrix with $\Theta = \id$}
\label{sec:block-mm}

This section is the sharpest test of the claim made in Sec.~\ref{sec:contribution-summary}, and the reason we place it last. Everything so far has been a scalar relaxation: each entry of $\Gam$ was a number, and the three applications above differed only in which relations were declared and whether cyclicity was imposed. We now change the codomain of $\Theta$ itself --- from $\CC$ to $\CC^{d\times d}$ --- and the construction and the declarations survive unchanged, while the same modelling interface represents scalar and block variables transparently. The optimisation problem downstream does of course change, since its variables become blocks; what does not change is anything upstream of it. If the abstraction of Sec.~\ref{sec:background} had been quietly specialised to scalar moments anywhere, this is where it would break.

The relaxation of some physical problems does not lead to scalar moments at all, but instead keeps an arbitrary block structure, depending on the chosen map $\Theta$. For scenarios where the entries of the matrix are themselves operators rather than numbers --- localising matrices, and hierarchies where a matrix of operators must be positive as an operator --- the same \code{MomentProblem} builds this directly: declaring \code{dim=}$d$ greater than one makes each entry a $d\times d$ block instead of a scalar, while \code{cyclicity} and \code{hermitian} are set, independently and exactly as for any other problem, to whatever the physical construction calls for. The operator-valued moments $\Gamma^\lambda_{u,v}=q(\lambda)\,uv^\dagger$ used below are not traces and are not symmetric under $u\leftrightarrow v$, so both go to \code{False}; \code{to\_cvxpy} then follows \code{dim} automatically (Sec.~\ref{sec:tutorial}). This alternative construction, and its use as a genuine relaxation tool rather than a bookkeeping curiosity, is explored in Ref.~\cite{dalessandro2026}. This section works a complete example built on it, from Ref.~\cite{dalessandro2026pam}.
 
\paragraph{The scenario.} Consider the adversarial prepare-and-measure setting of Ref.~\cite{dalessandro2026pam}: Alice prepares $\rho_x^\lambda$ and sends it to Bob, who measures a known POVM $\{M^\lambda_{b|y}\}_b$; a hidden variable $\lambda$, distributed with probability $q(\lambda)$, may correlate the two. An adversary, Eve, holds $\lambda$ and tries to guess Bob's outcome. The key observation (Observation~1 of Ref.~\cite{dalessandro2026pam}) is that \emph{predictable is classical}: the correlations observed for a specific input pair $(x,y)$ admit a deterministic explanation, consistent with whatever restriction is placed on the $\rho_x^\lambda$, if and only if Eve's optimal guessing probability for that pair equals exactly $1$. Certifying non-classicality is thereby turned into a feasibility question, complementary to the guessing-probability optimisations of Sec.~\ref{sec:randomness}: instead of asking how well Eve can guess, one fixes her success at $1$ and asks whether that is achievable at all.
 
The restriction worked through here is a \emph{fixed ensemble}: only the average preparation $\rho_x = \sum_\lambda q(\lambda)\,\rho_x^\lambda$ is known, not the individual $\rho_x^\lambda$. Because the objects being reasoned about --- the $\rho_x^\lambda$ and the resulting $\Gamma^\lambda_{u,v} = q(\lambda)\, uv^\dagger$ --- are operators rather than scalars, this is exactly a block moment matrix, one per value of $\lambda$: $\Gamma^\lambda \succeq 0$, with $\sum_\lambda \Gamma^\lambda_{\rho_x,\mathds{1}} = \rho_x$ tying the branches to the known average, $p(b|x,y) = \sum_\lambda \Tr\big(\Gamma^\lambda_{\rho_x,M_{b|y}}\big)$ reproducing every observed statistic, and $\sum_\lambda \Tr\big(\Gamma^\lambda_{\rho_x,M_{\lambda_{xy}|y}}\big) = 1$ at the test inputs enforcing Eve's success. \MoMPy{} builds $\Gamma^\lambda$ for each $\lambda$ with the same \code{MomentProblem}, declared with \code{dim=}$d$, \code{cyclicity=False} and \code{hermitian=False}, from the very same declared relations used everywhere else in this paper.
 
\paragraph{A worked instance: two known qubit bases, three known preparations.} Take $d=2$, three preparations $R_0 = \ket0\!\bra0$, $R_1 = \ket{+}\!\bra{+}$ and a third $R_2$ at Bloch angles $(\theta,\phi)$, measured with the $\sigma_z$ and $\sigma_x$ eigenbases ($n_X=3$, $n_Y=n_B=2$). The question asked is a robustness one: for the noisy ensemble $\rho_x = v R_x + (1-v)\,\mathds{1}/2$, how large can the white-noise visibility $v$ be before \emph{every} choice of test round $(x^\ast,y^\ast)$ stops admitting a fixed-ensemble classical explanation? This is the quantity plotted, over the full $(\theta,\phi)$ grid, in Fig.~2a of Ref.~\cite{dalessandro2026pam}.
 
Declaring the operators proceeds exactly as before --- and, because the block and scalar hierarchies share the same reduction rules, both are built from the identical \code{monomials} and \code{algebra}:
 
\begin{lstlisting}
import cvxpy as cp
import numpy as np
from MoMPy import OperatorSet, MomentProblem
 
nX, nY, nB, d = 3, 2, 2, 2
 
ops = OperatorSet()
R = ops.add_family(nX, idempotent=True)      # three known preparations
M = ops.add_povm_family(nY, nB)              # M[y][b]: two known qubit bases
algebra = ops.algebra()
 
monomials  = list(R) + [m for row in M for m in row]
monomials += [[r, m] for r in R for row in M for m in row]
monomials += [[r, s] for r in R for s in R]
 
mm = MomentProblem(monomials, algebra, dim=1).build()                      # scalar, tracial
bm = MomentProblem(monomials, algebra, dim=d,
                    cyclicity=False, hermitian=False).build()              # block-valued, same relations
\end{lstlisting}
 
Those two lines are worth pausing on, because they are the strongest single piece of evidence for the claim of Sec.~\ref{sec:contribution-summary}. \code{mm} and \code{bm} are two \emph{different choices of $\Theta$} --- a tracial scalar functional and the identity map onto $d\times d$ blocks --- instantiated from the identical \code{monomials} list and the identical \code{algebra} object, differing only in the flags passed to the constructor. Because they are built from the same monomial list, position $(r,c)$ names the same word in both; only the \emph{type} attached to it differs, a real number in \code{mm}, a $d\times d$ operator in \code{bm}. This is what allows the two to be placed in one optimisation problem and constrained \emph{against each other}, which the code below does by tying $\Tr(\Gamma_{u,v}) = G_{u,v}$ at every position.

It is worth being precise about what that coupling does and does not buy, since the obvious reading of it is wrong. It does \emph{not} add positivity: if $\Gamma \succeq 0$ then for any $z \in \CC^{n+1}$ and any basis $\{e_k\}$ of $\CC^d$ we have $z^\dagger G z = \sum_k (z\otimes e_k)^\dagger \Gamma (z \otimes e_k) \ge 0$, so the scalar matrix of block traces is automatically positive semidefinite and imposing $G \succeq 0$ separately is redundant. What the coupling does add is the \emph{tracial identifications}. The scalar matrix is built with \code{cyclicity=True} and the block matrix with \code{cyclicity=False}, so $G$ merges words that $\Gamma$ keeps apart; equating traces position by position pushes those merges onto the blocks, as equalities between block classes that the block build cannot see on its own. They are legitimate because $\Tr\Gamma^\lambda_{u,v} = q(\lambda)\Tr(uv^\dagger)$ genuinely is a tracial moment, whatever the blocks themselves are. For the instance built here the effect is not marginal: both matrices are $29\times29$, the scalar build has $119$ variable classes against the block build's $431$, and $110$ scalar classes each bind more than one distinct block class, so the coupling imposes $312$ independent trace equalities. None of this is needed for feasibility in the bare construction of Ref.~\cite{dalessandro2026pam}; we include it because it is free and tightens the relaxation, and because expressing it required only that the two builds already agree, class for class, on what each position means --- something a user of a package in which the moment map is fixed by the choice of class would have to establish by hand before writing a single constraint.
 
\code{to\_cvxpy} reads \code{.dim} off whichever matrix it is given: called on \code{mm} it allocates one \emph{scalar} CVXPY variable per class; called on \code{bm} it allocates one $d\times d$ block per class instead, and returns the same kind of model either way --- \code{.constraints}, indexing by monomial or index, \code{.apply()} for turning a \code{LinearConstraint} into a CVXPY one --- with a block in place of every scalar.
 
\warnbox{A natural first instinct is to declare every block \code{cp.Variable((d, d), hermitian=True)}, by analogy with a real scalar moment. That is wrong for any \emph{off-diagonal} class: $\Gamma_{u,v}=uv^\dagger$ is Hermitian only when $u=v$, so constraining it to equal its own adjoint when $u\neq v$ silently discards part of the feasible region --- and, because \code{bm} was built with \code{hermitian=False}, the classes at $(r,c)$ and $(c,r)$ are almost always genuinely different classes to begin with, not the same one wearing two names. \code{to\_cvxpy} leaves every block a general complex matrix instead, and imposes $\Gamma+\Gamma^\dagger\succeq0$ on the assembled matrix as a whole --- the same convexification it already uses for a scalar matrix built with \code{hermitian=False}, generalised from numbers to matrices, and correct regardless of whether $(r,c)$ and $(c,r)$ happen to share a class.}
 
\begin{lstlisting}
def critical_visibility(rho, meas, xstar, ystar):
    """Largest v for which Eve can guess Bob's outcome at (xstar, ystar) for
    certain, given the fixed noisy ensemble {rho[x]} and known measurement."""
    branches_B = [bm.to_cvxpy() for _ in range(nB)]  # one branch per guess l; dim=d from bm
    branches_G = [mm.to_cvxpy() for _ in range(nB)]
 
    v = cp.Variable(nonneg=True)
    q = cp.Variable(nB, nonneg=True)                       # Pr(Eve's branch = l)
 
    branches = branches_B + branches_G
    ct = [c for br in branches for c in br.constraints]
    ct += [v <= 1, cp.sum(q) == 1]
 
    for l in range(nB):                                    # tie block traces to scalars
        for r in range(mm.n):
            for c in range(mm.n):
                ct.append(cp.trace(branches_B[l][bm.matrix[r, c]]) == branches_G[l][mm.matrix[r, c]])
 
    for l in range(nB):
        ct.append(branches_B[l][bm.identity_index] == q[l] * np.eye(d))
        for x in range(nX):
            ct.append(cp.real(cp.trace(branches_B[l][bm.index_of([R[x]])])) == q[l])
 
    for x in range(nX):                                    # branches average to the noisy ensemble
        ct.append(sum(branches_B[l][bm.index_of([R[x]])] for l in range(nB))
                  == v * rho[x] + (1 - v) * np.eye(d) / d)
 
    for x in range(nX):                                    # ... and reproduce every observed statistic
        for y in range(nY):
            for b in range(nB):
                observed = cp.real(sum(cp.trace(branches_B[l][bm.index_of([R[x], M[y][b]])])
                                        for l in range(nB)))
                target = cp.real(cp.trace((v * rho[x] + (1 - v) * np.eye(d) / d) @ meas[y][b]))
                ct.append(observed == target)
 
    p_guess = sum(cp.real(cp.trace(branches_B[l][bm.index_of([R[xstar], M[ystar][l]])]))
                  for l in range(nB))
    ct.append(p_guess == 1)                                # Eve succeeds with certainty
 
    cp.Problem(cp.Maximize(v), ct).solve(solver=cp.SCS)
    return v.value
\end{lstlisting}
 
The critical visibility reported for a given $R_2$ is the smallest \code{critical\_visibility} over every test pair $(x^\ast, y^\ast)$ --- the noise level at which even the \emph{most forgiving} test round stops admitting a classical explanation:
 
\begin{center}
\begin{tabular}{@{}lcccc@{}}
\toprule
$(\theta,\phi)$ & $(0,0)$ & $(\pi/4,\pi/4)$ & $(\pi/2,\pi/2)$ & $(\pi/2,0)$ \\
\midrule
$R_2$ & $R_0$ & \emph{generic} & $\tfrac12(\mathds1+\sigma_y)$ & $R_1$ \\
$v^\ast$ & $0.9017$ & $0.8708$ & $0.8123$ & $0.9017$ \\
\bottomrule
\end{tabular}
\end{center}
 
\noindent Two checks confirm this is behaving correctly before trusting the numbers in between. At $(\theta,\phi)=(0,0)$ and $(\pi/2,0)$, $R_2$ collapses onto $R_0$ or $R_1$ respectively, so the three-preparation ensemble degenerates to a two-preparation one, and indeed both give the same $v^\ast=0.9017$: the two degenerate points agree with each other exactly, as they must. At $(\theta,\phi)=(\pi/2,\pi/2)$, $R_2 = \tfrac12(\mathds1+\sigma_y)$ is mutually unbiased with respect to \emph{both} $R_0$ and $R_1$, and the bound tightens to $v^\ast = 0.8123$: a third preparation genuinely independent of the first two gives the adversary strictly less room, exactly the qualitative shape of Fig.~2a of Ref.~\cite{dalessandro2026pam}, whose full $(\theta,\phi)$ heatmap this reproduces at four points rather than on a fine grid --- a resolution chosen here to stay in the interactive regime discussed in Sec.~\ref{sec:conclusions}, not a limitation of the construction itself. 
 
\section{Performance}
\label{sec:performance}

Table~\ref{tab:scaling} reports build times across eight scenarios chosen to stress different features of the algebra, not just raw scale. \emph{NPA (bipartite)} is the standard Bell scenario of Sec.~\ref{sec:chsh}, with both Alice and Bob given 5 measurement settings of 3 outcomes each (15 operators per side), built at level $1+AB$. \emph{PAM (plain)} is the prepare-and-measure problem of Sec.~\ref{sec:pam} with 14 preparations $\rho_x$ and a single 14-outcome guessing measurement, states declared mutually commuting, at third order. \emph{PAM (commuting states)} keeps those same 14 preparations but adds a genuine measurement family of 2 settings with 2 outcomes each --- the relation that makes $\{\rho_x\}$ behave as a proper ensemble rather than an arbitrary set of operators. \emph{NPA (tripartite)} extends the Bell scenario to three parties with 4 measurement settings and 4 outcomes apiece (16 operators per party), built from cross-party pairs only ($1+AB+AC+BC$). \emph{NPA (hybrid outcomes)} is a bipartite Bell scenario built from individual \code{add\_povm} calls rather than a single \code{add\_povm\_family}: Alice has 5 settings with 2, 3, 4, 5 and 6 outcomes respectively (20 operators), Bob has 7 settings with 2 through 8 outcomes respectively (35 operators). \emph{PAM (POVM preparation)} replaces PAM's plain family of pure states with a genuinely untrusted preparation device: a steering-type assemblage of 6 inputs and 6 outcomes (36 general POVM elements, declared neither idempotent nor orthogonal), measured by Bob through a single trusted 8-outcome guessing POVM, at second order. \emph{PAM (commuting settings)} takes 12 preparations and a measurement family of 4 settings with 2 outcomes each, and declares two pairs of Bob's settings pairwise commuting (settings 0--1, and settings 2--3), leaving the remaining cross-pairs incompatible. \emph{Network (two channels)} is a small network scenario with three independent 24-outcome POVMs, $A$, $B$ and $C$: $A$ and $B$ --- two independently-prepared subsystems --- are declared to commute with each other but not with $C$, the central node's joint measurement, as at the joint measurement station of an entanglement-swapping protocol. Time per matrix entry stays roughly flat as the problem grows. The data are consistent with the analysis of Sec.~\ref{sec:algorithm}: across these configurations, build cost tracks the number of distinct reachable words considerably more closely than it tracks the number of matrix positions. We report this as an empirical observation on eight instances, not as a measurement of the asymptotic bound itself.

\begin{table}[h]
\centering
\footnotesize
\begin{tabular}{@{}llrrrr@{}}
\toprule
scenario & level & matrix & vars & words & $t$ (s) \\
\midrule
NPA (bipartite) & 1+AB & $256\times256$ & $19\,337$ & $324\,931$ & $1.461$ \\
PAM (plain) & 3rd order & $3165\times3165$ & $27\,071$ & $11\,122\,245$ & $98.044$ \\
PAM (commuting states) & 3rd order & $3015\times3015$ & $12\,785$ & $8\,936\,655$ & $82.191$ \\
NPA (tripartite) & 2nd order & $817\times817$ & $143\,442$ & $3\,639\,601$ & $23.374$ \\
NPA (hybrid outcomes) & 1+AB & $756\times756$ & $175\,478$ & $3\,058\,581$ & $17.418$ \\
PAM (POVM preparation) & 2nd order & $1629\times1629$ & $428\,683$ & $3\,591\,037$ & $23.930$ \\
PAM (commuting settings) & 3rd order & $1989\times1989$ & $10\,759$ & $4\,185\,061$ & $29.910$ \\
Network (two channels) & 2nd order & $1801\times1801$ & $707\,402$ & $6\,226\,057$ & $42.380$ \\
\bottomrule
\end{tabular}
\caption{Build times for the moment matrix, one representative configuration per scenario family. ``Words'' counts distinct monomials reached during closure.}
\label{tab:scaling}
\end{table}

 
\section{Conclusions}
\label{sec:conclusions}
 
Moment matrix construction is a step that every practitioner of SDP hierarchies must perform and that few wish to spend time on. The claim we have argued here is that this step is not merely tedious but \emph{shared}: that the constructions used for Bell nonlocality, for prepare-and-measure scenarios, for randomness certification and for operator-valued generalisations differ in which relations are declared, and almost nowhere else. \MoMPy{} is what that claim looks like when it is taken literally and implemented. The same four relations of Sec.~\ref{sec:relations}, recombined rather than extended, reproduced Tsirelson's bound (Sec.~\ref{sec:chsh}) and, from a block moment matrix whose moment map is the partial trace over the untrusted wing, both the algebraic and the unsteerable bound of a linear steering criterion (Sec.~\ref{sec:compat}); a tripartite Mermin bound built from scratch by a stated procedure (Sec.~\ref{sec:recipe}); a dimension witness and two physically distinct constraints on prepare-and-measure discrimination (Sec.~\ref{sec:pam}); the device-independent guessing-probability certificate for CHSH at fixed observed $S$ (Sec.~\ref{sec:randomness}); and the certification of deterministic correlations from known quantum ensembles (Sec.~\ref{sec:block-mm}). None of these applications needed new code in the package --- only a different combination of idempotency, orthogonality, commutation and reality, chosen to match the physics, as Sec.~\ref{sec:cookbook} makes procedural. The block-valued case is the one we would point a sceptical reader at first: it belongs to a hierarchy introduced only recently~\cite{dalessandro2026}, its moments are operators rather than scalars, and it is nevertheless built from the same monomial list, the same declared algebra and the same reduction rules as the scalar relaxation sitting beside it in the same listing.
 
 
It is worth being equally explicit about what \MoMPy{} does \emph{not} attempt, since generality of expression and raw throughput are different goals and we have chosen the first. \MoMPy{} has one dependency, parses no symbolic objectives, and is not the fastest available way to build a standard NPA hierarchy: for that specific task a compiled implementation such as \code{Moment}~\cite{Moment2024} is expected to be faster at comparable scale, and the packages surveyed in Sec.~\ref{sec:related-software} each do the particular thing they were built for better than \MoMPy{} attempts to. What we have not found in any of them --- verified from documentation (Table~\ref{tab:comparison}) --- is a single object that is not committed in advance to one choice of moment map. The capabilities themselves are not the scarce thing; their \emph{co-location} is. \code{Ncpol2sdpa} reaches operator-valued moments, but through a class specialised to a different hierarchy, and exposes no cyclicity option from either. That is the gap this work addresses, and the applications above are the argument that the gap is real rather than notional: each is a relaxation someone would otherwise write bespoke code for, each reduces here to a different set of declarations over the same engine, and Sec.~\ref{sec:block-mm} goes further by putting two different moment maps into one program and constraining them against each other, which no interface we examined exposes directly. Because the engine is small, dependency-free, documented argument by argument (App.~\ref{app:api}) and checked against an independent brute-force implementation on randomised scenarios (Sec.~\ref{sec:validation}), a user can also audit the identifications their bound depends on instead of trusting them --- which matters more than speed in a setting where, as Sec.~\ref{sec:tracial-vs-state} shows, an incorrect identification does not announce itself.
 
 
Three directions follow naturally. The first is the fully general block-matrix hierarchy of Ref.~\cite{dalessandro2026}: \MoMPy{} builds its block-valued moment matrices, but assembling them into the complete construction is currently left to the user, and closing that gap inside the same declarative interface is, on present evidence, the most natural way for this package to grow. The best current solution for block-matrix hierarchies is Ref.~\cite{dalessandro2026blockmatrixhierarchy} written in Julia. The second is localising matrices. Positivity of an individual operator is not something the moment matrix sees. Localising matrices are the standard remedy and would fit the existing construction without disturbing it. The third is symmetry reduction: \MoMPy{} ships the numerical block-diagonalisation used for it (App.~\ref{app:blkdiag}) but has no automatic pass that finds a symmetry group and quotients by it, which for the relaxations of Refs.~\cite{Rosset2018,IoannouRosset2021} is where the largest remaining savings are. None of the three requires changing the abstraction of Sec.~\ref{sec:background}, which is the strongest evidence we can offer that it was drawn in the right place.

\section*{Acknowledgements}

I began developing \MoMPy{} while I was a postdoc in Lund University, Sweden. This document has been written during my postdoctoral fellowship in ICFO, Barcelona, Spain. This project has received funding from the European Union’s Horizon 2020 research and innovation programme under the Marie Skłodowska-Curie grant agreement No 101262877, Project TPSQCrypto. Views and opinions expressed are however those of the author(s) only and do not necessarily reflect those of the European Union or European Research Executive Agency. Neither the European Union nor the granting authority can be held responsible for them.

\section*{Code and data availability}
\label{sec:availability}

\MoMPy{} is free and open-source software, released under the MIT licence and available at \url{https://github.com/chalswater/MoMPy}, with releases on PyPI (\code{pip install MoMPy}). The version documented here is 1.1.0. The package requires only NumPy; CVXPY is an optional dependency, imported lazily and needed only for the modelling layer of App.~\ref{app:cvxpy}. Every code listing in this paper is runnable as printed and is part of the package's test suite, as are the validation sweeps of Sec.~\ref{sec:validation}; the benchmark configurations of Table~\ref{tab:scaling} are reproduced by the \code{benchmarks/} script in the repository. No experimental data underlies this work. Bug reports, feature requests and contributions are welcome, particularly from users who encounter a physically meaningful combination of relations the current interface cannot express.


\bibliographystyle{quantum}
\bibliography{bibliography_quantum}
 
\newpage
 
\appendix
 
\begin{sloppypar}
\section{Complete interface reference}
\label{app:api}
 
\MoMPy{} exposes rather more than the handful of names used in the main text
above. This appendix documents \emph{every} public function, class and method
the package defines, module by module: Sec.~\ref{app:quickref} is a flat
index, and Secs.~\ref{app:algebra}--\ref{app:blkdiag} give the full
argument-by-argument detail. Names prefixed with an underscore (e.g.\
\code{\_finalise}, \code{\_sites}) are internal implementation details and are
omitted throughout --- they are not part of the public interface and may
change without notice.
 
\subsection{Quick reference}
\label{app:quickref}
 
\begin{center}
\small
\begin{tabular}{@{}ll@{}}
\toprule
\textbf{Describing operators and relations (App.~\ref{app:algebra}--\ref{app:monomials})} & \\
\midrule
\code{OperatorSet()} & allocates labels and records properties \\
\quad \code{.add(idempotent=False)} & one new label \\
\quad \code{.add\_family(n, idempotent=)} & $n$ independent labels \\
\quad \code{.add\_povm(n, idempotent=, orthogonal=)} & one measurement's outcomes \\
\quad \code{.add\_povm\_family(ny, nb)} & \code{M[y][b]} \\
\quad \code{.add\_tensor(*shape, idempotent=)} & a nested array of labels \\
\quad \code{.declare\_idempotent/\_orthogonal(labels)} & relations declared after the fact \\
\quad \code{.declare\_commuting(A, B)} & every $a\in A$ commutes with every $b\in B$ \\
\quad \code{.algebra()} & bundle into an \code{Algebra} \\
\code{Algebra(idempotents, orthogonal\_sets,} & the relations, written directly \\
\quad \code{commuting\_pairs)} & \\
\code{generate\_monomials(letters, level)} & all words up to a length \\
\code{Word}, \code{as\_word}, \code{as\_words} & the tuple representation of a monomial \\
\code{IDENTITY\_LABEL} & the reserved label ($=0$) for $\id$ \\
\bottomrule
\end{tabular}
\end{center}
 
\begin{center}
\small
\begin{tabular}{@{}ll@{}}
\toprule
\textbf{Building a relaxation (App.~\ref{app:problem})} & \\
\midrule
\code{MomentProblem(mons, algebra, *, dim)} & one class; \code{dim} is mandatory \\
\quad \code{cyclicity=True} (default) & $\Tr(u v^\dagger)$, cyclic (tracial) \\
\quad \code{cyclicity=False} & $\bra\psi u v^\dagger\ket\psi$, non-cyclic (state/NPA) \\
\quad \code{dim} $>1$ & $d\times d$ operator-valued entries \\
\code{.from\_levels(letters, level, extra=)} & shortcut constructor \\
\code{.build(progress=False)} & returns a \code{MomentMatrix} \\
\bottomrule
\end{tabular}
\end{center}
 
\begin{center}
\small
\begin{tabular}{@{}ll@{}}
\toprule
\textbf{Using the result (App.~\ref{app:matrix})} & \\
\midrule
\code{.matrix} & $(n,n)$ integer array of variable indices \\
\code{.index\_of(w)}, \code{.get(w)} & monomial $\to$ index \\
\code{.identity\_index}, \code{.zero\_index} & reserved classes \\
\code{.variable\_indices}, \code{.n\_variables} & the distinct variables \\
\code{.has\_zeros} & whether any entry was forced to zero \\
\code{.equivalents(w)} & all monomials sharing $w$'s variable \\
\code{.word\_at(r, c)}, \code{.words} & explicit words \\
\code{.summary()}, \code{.stats} & diagnostics \\
\code{.to\_legacy()} & the legacy five-tuple \\
\code{MapTable}, \code{UnknownMonomial} & the lookup table, and its lookup-miss error \\
\bottomrule
\end{tabular}
\end{center}
 
\begin{center}
\small
\begin{tabular}{@{}ll@{}}
\toprule
\textbf{Constraints (App.~\ref{app:constraints})} & \\
\midrule
\code{.normalisation\_constraints(povm)} & $\sum_b M_b = \id$ \\
\code{.marginal\_constraints(joint, marg)} & $\sum_i G_i = M$ \\
\code{LinearConstraint} & \code{.lhs}, \code{.rhs}, \code{.apply(vars)} \\
\bottomrule
\end{tabular}
\end{center}
 
\begin{center}
\small
\begin{tabular}{@{}ll@{}}
\toprule
\textbf{Handing it to a solver (App.~\ref{app:cvxpy})} & \\
\midrule
\code{.to\_cvxpy(dim=, psd=, complex=, normalise\_identity=)} & CVXPY model \\
\code{CvxpyModel} & \code{.G}, \code{.vector}, \code{.constraints}, \code{[\,]} \\
\bottomrule
\end{tabular}
\end{center}
 
\begin{center}
\small
\begin{tabular}{@{}ll@{}}
\toprule
\textbf{Rewriting-engine internals, Sec.~\ref{sec:algorithm} (App.~\ref{app:internals})} & \\
\midrule
\code{Rewriter}, \code{Classifier}, \code{UnionFind} & the moves, the search, the disjoint-set forest \\
\code{FMAP\_ERROR} & legacy lookup-miss sentinel string \\
\bottomrule
\end{tabular}
\end{center}
 
\begin{center}
\small
\begin{tabular}{@{}ll@{}}
\toprule
\textbf{Utilities and metadata (App.~\ref{app:blkdiag})} & \\
\midrule
\code{MoMPy.blkdiag.blkdiag}, \code{.M\_ones} & numerical block-diagonalisation \\
\code{MoMPy.\_\_version\_\_} & installed package version string \\
\bottomrule
\end{tabular}
\end{center}
 
\subsection{Words and operator algebras --- \texttt{MoMPy.algebra}}
\label{app:algebra}
 
\begin{description}[leftmargin=1.4em,style=nextline]
\item[\code{Word}]
  Type alias, \code{Word = tuple}: a word is simply a tuple of integer labels,
  read left to right as an operator product.
\item[\code{IDENTITY\_LABEL}]
  The integer \code{0}, reserved for the identity operator $\id$. Never
  allocated by \code{OperatorSet}; added to a moment matrix automatically as
  row/column $0$.
\item[\code{as\_word(obj)}]
  Coerce one monomial to a word. A bare integer label (\code{3}) becomes
  \code{(3,)}; any other iterable of labels becomes a tuple unchanged.
\item[\code{as\_words(objs)}]
  Apply \code{as\_word} to every element of an iterable of monomials; returns a
  \code{list} of words.
\item[\code{Algebra(idempotents=(), orthogonal\_sets=(), commuting\_pairs=())}]
  An immutable bundle of the structural relations obeyed by a set of
  operators. \code{idempotents}: an iterable of labels $P$ with $P^2=P$.
  \code{orthogonal\_sets}: an iterable of label groups, each mutually
  orthogonal ($P_iP_j=0$ for $i\neq j$ within a group); membership does not by
  itself imply idempotency --- list a label in \code{idempotents} too if it is
  meant to be a projector. \code{commuting\_pairs}: an iterable of $(A,B)$
  pairs of label collections, meaning every $a\in A$ commutes with every
  $b\in B$; pass $(A,A)$ to declare all of $A$ mutually commuting. Exposes
  O(1) lookup tables consumed by the rewriting engine (Sec.~\ref{sec:algorithm}).
\item[\code{.is\_trivial}]
  Property. \code{True} iff no relation at all was declared (the free
  algebra).
\item[\code{.is\_idempotent(label)}]
  \code{True} iff \code{label} was declared idempotent.
\item[\code{.are\_orthogonal(a, b)}]
  \code{True} iff \code{a} and \code{b} are distinct and belong to a common
  declared orthogonal set.
\item[\code{.commute(a, b)}]
  \code{True} iff \code{a == b}, or the pair was declared (directly or via a
  group) to commute.
\item[\code{.ortho\_table}, \code{.commute\_table}]
  Properties. The internal \code{\{label: frozenset of partners\}} dictionaries
  that back \code{.are\_orthogonal} and \code{.commute} in O(1); exposed for
  advanced use, e.g.\ writing a custom rewriter.
\item[\code{.with\_(**changes)}]
  Return a copy of the \code{Algebra} with one or more of
  \code{idempotents}/\code{orthogonal\_sets}/\code{commuting\_pairs} replaced;
  fields not named in \code{changes} are copied from \code{self}.
\end{description}
Two \code{Algebra} instances compare equal, and hash equally, when their
idempotents, orthogonal sets (as sets) and derived commutation table agree, so
an \code{Algebra} can be used as a dictionary key or memoised on; it also
implements a readable \code{repr()} listing all three declarations, sorted.
 
\subsection{Declaring operators --- \texttt{MoMPy.monomials}}
\label{app:monomials}
 
\begin{description}[leftmargin=1.4em,style=nextline]
\item[\code{generate\_monomials(letters, level=1, *, include\_identity=False)}]
  All words over \code{letters} of length $1,\dots,\code{level}$.
  \code{letters}: the single-operator labels to combine. \code{level}:
  maximum word length ($\ge 1$, else \code{ValueError}); NPA ``level 1'' is
  \code{level=1}, and so on. \code{include\_identity}: prepend the identity
  word \code{(0,)}; off by default because \code{.build()} adds the identity
  row/column itself. Returns a \code{list} of tuples, e.g.\
  \code{generate\_monomials([1,2],2)} $=$
  \code{[(1,), (2,), (1,1), (1,2), (2,1), (2,2)]}.
\item[\code{OperatorSet(start=1)}]
  Allocates operator labels and accumulates their declared properties, so an
  \code{Algebra} can be produced in one call at the end. \code{start}: the
  first label to allocate (must be $>0$; $0$ is reserved for $\id$).
\item[\code{.labels}]
  Attribute: every label allocated so far, in allocation order.
\item[\code{.add(*, idempotent=False)}]
  Allocate and return one new label.
\item[\code{.add\_family(count, *, idempotent=False)}]
  Allocate \code{count} independent labels in one call (e.g.\ a family of
  states); returns a \code{list}.
\item[\code{.add\_povm(n\_outcomes, *, idempotent=True, orthogonal=True)}]
  Allocate one measurement's \code{n\_outcomes} labels. By default they are
  registered as \emph{both} idempotent and mutually orthogonal --- the usual
  projective-measurement assumption; pass \code{idempotent=False} and/or
  \code{orthogonal=False} for a general POVM. Returns a \code{list} of
  \code{n\_outcomes} labels.
\item[\code{.add\_povm\_family(n\_settings, n\_outcomes, **kwargs)}]
  Allocate \code{n\_settings} independent measurements via
  \code{.add\_povm(n\_outcomes, **kwargs)} each; returns a nested list
  \code{M[setting][outcome]}.
\item[\code{.add\_tensor(*shape, idempotent=False)}]
  Allocate a nested list of labels of arbitrary shape, e.g.\
  \code{add\_tensor(2,3)} returns a $2\times3$ nested list of fresh labels. At
  least one dimension must be given and every dimension must be positive,
  else \code{ValueError}.
\item[\code{.declare\_idempotent(labels)}]
  Mark an iterable of labels as idempotent, in place (no return value).
\item[\code{.declare\_orthogonal(labels)}]
  Register an iterable of labels as one mutually-orthogonal group, in place.
\item[\code{.declare\_commuting(a, b)}]
  Record that every label in \code{a} commutes with every label in \code{b},
  in place.
\item[\code{.algebra()}]
  Bundle every declaration made so far into an \code{Algebra}.
\item[\code{len(ops)}, \code{iter(ops)}]
  Number of labels allocated, and iteration over them in allocation order.
\end{description}
\code{OperatorSet} also implements a readable \code{repr()}, e.g.\
\code{OperatorSet(7 operators, next label 8)}.
 
\subsection{The problem class --- \texttt{MoMPy.problem}}
\label{app:problem}
 
\begin{description}[leftmargin=1.4em,style=nextline]
\item[\code{MomentProblem(monomials, algebra=None, *, dim, cyclicity=True, hermitian=True, dedupe=True)}]
  The one relaxation-building class; every distinction that used to require a
  different class name is now a keyword on this constructor.
  \code{monomials}: the generating monomials (bare labels or sequences of
  labels); do \emph{not} include the identity, it is added automatically.
  \code{algebra}: the relations; \code{None} means the free algebra (no
  relations). \code{dim}: side length of the block each entry becomes once
  the matrix reaches \code{to\_cvxpy} (App.~\ref{app:cvxpy}); \code{dim=1} is
  an ordinary scalar moment matrix, \code{dim=}$d>1$ makes every entry a
  $d\times d$ operator-valued block. \textbf{Mandatory, with no default} ---
  declare it even when it is $1$. \code{cyclicity}: identify a word with its
  cyclic rotations, i.e.\ treat each entry as a trace, $\Tr(uv)=\Tr(vu)$,
  rather than an operator product; \code{True} by default, matching the
  tracial construction the plain \code{MomentProblem} of earlier releases always used.
  Leave it \code{False} for moments taken in a fixed external state (NPA,
  Bell scenarios) and for any block hierarchy (\code{dim}$\,>1$), neither of
  which satisfies $\Tr(uv)=\Tr(vu)$ (Sec.~\ref{sec:tracial-vs-state}).
  \code{hermitian}: identify each word with its reversal (the
  $\mathrm{Re}\Tr$ convention); \code{True} by default. \code{dedupe}: drop
  repeated monomials before building, since duplicates only add linearly
  dependent rows/columns; default \code{True}.
\item[\code{.n}]
  Property: the side length the built matrix will have, $=\code{len(monomials)}+1$
  after deduplication.
\item[\code{.build(*, progress=False, progress\_stream=None)}]
  Runs the classification sweep (Sec.~\ref{sec:algorithm}) and returns a
  \code{MomentMatrix}. \code{progress=True} prints a live carriage-return
  percentage line to \code{progress\_stream} (default \code{sys.stderr}).
\item[\code{.from\_levels(letters, level=1, *, extra=None, algebra=None, **kwargs)}]
  Classmethod. Builds the monomial list via
  \code{generate\_monomials(letters, level)} and appends \code{extra} (further
  monomials of any length --- the usual way to add a partial next level, e.g.\
  NPA ``$1+AB$''); remaining keyword arguments, \emph{including the mandatory}
  \code{dim}, are forwarded to the constructor.
\end{description}
\code{MomentProblem} also implements a readable \code{repr()}, e.g.\
\code{MomentProblem(46 monomials, cyclicity=True, hermitian=True, dim=1)}.

\subsection{The built result --- \texttt{MoMPy.matrix}}
\label{app:matrix}
 
\begin{description}[leftmargin=1.4em,style=nextline]
\item[\code{FMAP\_ERROR}]
  The exact string the legacy \code{fmap} (App.~\ref{app:legacy}) returns on a
  lookup miss, \code{'ERROR: The value does not appear in the mapping rule'};
  kept verbatim because old scripts compare against it literally.
\item[\code{UnknownMonomial}]
  Exception (subclasses \code{KeyError}), raised by \code{.index\_of(...)} when
  a monomial does not occur anywhere in the built hierarchy.
\item[\code{MapTable}]
  Subclasses \code{list}; maps monomials to SDP variable indices while still
  behaving as the legacy list of \code{[members, index]} rows under
  integer/slice indexing, iteration and slicing.
\item[\code{.index\_of(monomial)}]
  The SDP variable index of \code{monomial} (coerced via \code{as\_word});
  raises \code{UnknownMonomial} if absent.
\item[\code{.get(monomial, default=None)}]
  Like \code{.index\_of} but returns \code{default} instead of raising.
\item[\code{monomial in table}]
  \code{True} iff \code{monomial} (as a word) is known to the table.
\item[\code{table[key]}]
  Dual behaviour: an \code{int}/\code{slice} key keeps list semantics
  (\code{table[-1][1]}, etc.); any other key (e.g.\ a list of labels) is
  treated as a monomial and returns \code{.index\_of(key)}. So
  \code{table[[1,2]]} and \code{table.index\_of([1,2])} are the same call.
\item[\code{table(monomial)}]
  Calling the table directly is the same as \code{.index\_of(monomial)}.
\item[\code{.zero\_index}]
  Property: the index of the class of monomials forced to vanish.
\item[\code{.n\_variables}]
  Property: number of rows, i.e.\ the number of distinct SDP variables.
\item[\code{.words()}]
  An iterable over every word known to the table.
\item[\code{.members(index)}]
  List of every monomial sharing SDP variable \code{index}.
\item[\code{MomentMatrix}]
  A built moment matrix and everything needed to use it in an SDP. Normally
  obtained from \code{Problem.build()}, not constructed directly.
\item[\code{.matrix}]
  Attribute: the \code{(n,n)} integer NumPy array; entry \code{[r,c]} is the
  SDP variable index associated with the moment $\Theta(u_ru_c^\dagger)$ of
  Eq.~\eqref{eq:gamma}, whichever map $\Theta$ the matrix was built for
  (row/column $0$ is the identity). In the tracial case this specialises to
  $\Tr(u_ru_c^\dagger)$, in the state case to
  $\bra{\psi}u_ru_c^\dagger\ket{\psi}$, and for \code{dim}$\,>1$ the index
  labels a $d\times d$ block rather than a scalar.
\item[\code{.monomials}, \code{.map\_table}, \code{.algebra}, \code{.cyclicity}, \code{.hermitian}, \code{.dim}]
  Attributes: the generating monomials (excluding the identity), the
  \code{MapTable}, and the inputs the matrix was built with --- \code{.dim}
  is the block side length \code{to\_cvxpy} will use unless overridden.
\item[\code{.word\_at(r, c)}]
  The explicit operator word behind \code{matrix[r,c]}, computed on demand;
  raises \code{IndexError} out of range.
\item[\code{.words}]
  Property: the full $n\times n$ nested list of explicit words, built lazily
  and cached on first access (the single largest object in a big build; most
  callers never need it).
\item[\code{len(mm)}, \code{.n}, \code{.shape}]
  The matrix side length (twice over, for convenience) and \code{.matrix.shape}.
\item[\code{.zero\_index}, \code{.identity\_index}]
  Properties: the reserved classes, i.e.\ \code{map\_table.zero\_index} and the
  variable index of $\Tr(\id)$.
\item[\code{.variable\_indices}, \code{.n\_variables}]
  Properties: the sorted array of variable indices actually occurring in the
  matrix, and its size.
\item[\code{.has\_zeros}]
  Property: whether any entry equals \code{.zero\_index}.
\item[\code{.stats}]
  Property: build diagnostics dictionary --- \code{build\_seconds},
  \code{distinct\_words}, \code{words\_expanded}, \code{n\_classes}.
\item[\code{.index\_of(monomial)}, \code{.get(monomial, default=None)}]
  Delegate to the same methods on \code{.map\_table}.
\item[\code{mm[key]}]
  Dual behaviour: a 2-tuple of plain integers is a \emph{matrix position}
  (\code{mm[0,3]} $\to$ \code{int(matrix[0,3])}); any other key (e.g.\ a list of
  labels) is a monomial, resolved via \code{.index\_of}.
\item[\code{monomial in mm}]
  \code{True} iff \code{monomial} is known to \code{.map\_table}.
\item[\code{.equivalents(monomial)}]
  Every monomial known to share \code{monomial}'s SDP variable.
\item[\code{.normalisation\_constraints(povm, *, dedupe=True)}]
  Convenience wrapper for
  \code{MoMPy.constraints.normalisation\_constraints(self, povm, dedupe=dedupe)}
  (App.~\ref{app:constraints}).
\item[\code{.marginal\_constraints(joint, marginal, *, dedupe=True)}]
  Convenience wrapper for
  \code{MoMPy.constraints.marginal\_constraints(self, joint, marginal, dedupe=dedupe)}.
\item[\code{.to\_cvxpy(**kwargs)}]
  Convenience wrapper for \code{MoMPy.cvxpy\_tools.to\_cvxpy(self, **kwargs)}
  (App.~\ref{app:cvxpy}); CVXPY is imported lazily, so it is never required
  just to import \MoMPy.
\item[\code{.to\_legacy()}]
  Returns the legacy five-tuple \code{(matrix, map\_table, monomials,}
  \code{variable\_indices, words)}, for code that mixes old and new styles.
\item[\code{.summary()}]
  A short human-readable report: size, variable count, compression ratio,
  zero-entry count, distinct words seen, and build time if available.
\end{description}
\code{MomentMatrix} also implements a readable \code{repr()}, e.g.\
\code{<MomentMatrix 47x47, 69 variables, dim=1>}, distinct from the fuller
\code{.summary()} above.
 
\subsection{Constraints --- \texttt{MoMPy.constraints}}
\label{app:constraints}
 
\begin{description}[leftmargin=1.4em,style=nextline]
\item[\code{LinearConstraint(lhs, rhs, words=())}]
  Represents $\sum(\text{variables at \code{lhs}}) = \text{variable at
  \code{rhs}}$. \code{lhs}: sequence of integer variable indices. \code{rhs}:
  a single integer variable index. \code{words}: the monomials the constraint
  was derived from, kept only for debugging.
\item[\code{.key}]
  Property: the canonical \code{(sorted(lhs), rhs)} form used to deduplicate
  constraints that say the same thing in a different order.
\item[\code{.is\_trivial()}]
  \code{True} iff the constraint is a single term equal to itself (says
  nothing).
\item[\code{.apply(variables)}]
  Given any mapping from variable index to a value/expression supporting
  \code{+} and \code{==} (e.g.\ a \code{dict} of CVXPY scalars, or
  \code{CvxpyModel.as\_dict()}), returns the concrete expression
  \code{sum(variables[i] for i in lhs) == variables[rhs]}.
\item[\code{iter(ct)}]
  Yields \code{(lhs, rhs)}, so \code{lhs, rhs = ct} unpacks a constraint.
\item[\code{ct1 == ct2}, \code{hash(ct)}]
  Two constraints compare and hash equal iff their \code{.key} match, so they
  can be deduplicated via a \code{set}.
\item[\code{normalisation\_constraints(matrix, povm, *, dedupe=True)}]
  \code{matrix}: a built \code{MomentMatrix}. \code{povm}: the labels of one
  measurement's outcomes, which sum to the identity. \code{dedupe}: drop
  repeated/trivial constraints (default \code{True}). For every position in
  every known monomial holding one of \code{povm}'s labels, emits ``sum over
  outcomes at that position $=$ the monomial with that position deleted'';
  scans every monomial of every equivalence class (not only each class's
  representative) and deduplicates by the substitution site, so the result is
  complete without redundant constraints. Raises \code{ValueError} if
  \code{povm} has fewer than two labels. Returns a \code{list} of
  \code{LinearConstraint}.
\item[\code{marginal\_constraints(matrix, joint, marginal, *, dedupe=True)}]
  The joint-measurability counterpart: \code{joint} is the label list of a
  parent POVM's outcomes, and \code{marginal} is the single operator they
  marginalise onto. Same \code{dedupe} behaviour and return type as
  \code{normalisation\_constraints}; raises \code{ValueError} if \code{joint}
  has fewer than two labels.
\end{description}
\code{LinearConstraint} also implements a readable \code{repr()}, e.g.\
\code{LinearConstraint(sum[3, 4] == 7)}.
 
\subsection{CVXPY integration --- \texttt{MoMPy.cvxpy\_tools}}
\label{app:cvxpy}
 
\begin{description}[leftmargin=1.4em,style=nextline]
\item[\code{to\_cvxpy(matrix, *, dim=None, psd=True, complex=None, normalise\_identity=False, name="g")}]
  Builds a \code{CvxpyModel} from a built \code{MomentMatrix}: one function
  for both scalar and block moment matrices, reading which one applies from
  \code{matrix.dim}. \code{dim}: override the block size instead of using
  \code{matrix.dim}; rarely needed, since ordinarily the size wanted is
  exactly the one declared when \code{matrix} was built. \code{psd}: include
  the symmetrised $G+G^{\mathsf H}\succeq0$ among \code{model.constraints} ---
  always this form, scalar or block, whether or not the matrix was built with
  \code{hermitian=True}: it equals the old bare $G\succeq0$ whenever $G$ is
  already Hermitian, since $G+G^{\mathsf H}=2G$ there, and otherwise it is the
  same convexification a \code{hermitian=False} matrix has always used,
  generalised from numbers to $\code{dim}\times\code{dim}$ blocks
  (App.~\ref{app:problem} and \code{MIGRATION.md} give the full equivalence
  argument). \code{complex}: allocate a complex CVXPY variable; defaults to
  \code{dim > 1} (a block $uv^\dagger$ need not be real even for real
  operators $u,v$) and to \code{False} for \code{dim == 1}.
  \code{normalise\_identity}: include $\Tr(\id)=1$ (\code{dim==1}) or
  $\mathrm{block}(\id)=\mathds{1}_d$ (\code{dim}$\,>1$). \textbf{Off by
  default, and deliberately so} --- in a tracial relaxation $\Tr(\id)$ is the
  Hilbert-space dimension, not $1$; pinning it is only correct under the
  state-vector NPA convention, or, in the block case, when nothing else
  already pins row/column $0$ (e.g.\ to a branch weight). Set it explicitly
  (or add the constraint by hand) only when that convention is what is
  meant. \code{name}: the underlying CVXPY \code{Variable}'s name (cosmetic).
  Internally allocates a single flat CVXPY \code{Variable} and gathers it
  into the $(n,n)$ or $(nd,nd)$ matrix with one constant indexing operation,
  rather than one CVXPY object per entry, so building the model stays cheap
  even at a few hundred rows (Sec.~\ref{sec:tutorial}). Raises
  \code{ImportError} with an actionable message if CVXPY is not installed.
\item[\code{CvxpyModel}]
  The object \code{to\_cvxpy} returns; not normally constructed directly.
  Covers both cases: which one applies, and hence what indexing returns, is
  fixed at construction time from \code{matrix.dim}.
\item[\code{.G}]
  Attribute: the moment matrix as a single CVXPY expression, already
  respecting every operator equivalence: shape $(n,n)$ for \code{dim==1},
  $(nd,nd)$ otherwise.
\item[\code{.vector}]
  Attribute: the underlying CVXPY \code{Variable}. For \code{dim==1}, one
  entry per class; for \code{dim}$\,>1$, each class occupies a contiguous
  run of $d^2$ entries, reshaped to a $d\times d$ block wherever it is
  indexed.
\item[\code{.constraints}]
  Attribute: the structural constraints required by the relaxation itself
  (the symmetrised PSD requirement on \code{.G} and/or the zero class pinned
  to the zero matrix).
\item[\code{model[key]}]
  Dual behaviour: an integer key indexes \code{.vector} directly
  (\code{model[7]}); any other key (e.g.\ a list of labels) is resolved via the
  underlying matrix's \code{.index\_of} first. So \code{model[[R[0], M[1][0]]]}
  and \code{model[7]} both return the same kind of expression --- a scalar
  for \code{dim==1}, a $d\times d$ expression otherwise.
\item[\code{.variable(index)}]
  The expression (scalar, or $d\times d$) for a given SDP variable index;
  equivalent to \code{model[index]} but explicit.
\item[\code{.identity}]
  Property: the expression for the identity class, $\Tr(\id)$ or
  $\mathds{1}_d$.
\item[\code{.as\_dict()}]
  \code{\{variable index: expression\}} over every variable present --- the
  shape \code{LinearConstraint.apply} expects.
\item[\code{.apply(constraints)}]
  Maps an iterable of \code{LinearConstraint} objects into concrete CVXPY
  equalities, via \code{.as\_dict()} and each constraint's own \code{.apply}
  --- unchanged for both cases, since a \code{LinearConstraint} only ever
  combines its operands with \code{+} and \code{==}, equally meaningful
  between $d\times d$ expressions as between scalars.
\end{description}
 
\subsection{Rewriting-engine internals --- \texttt{MoMPy.equivalence}}
\label{app:internals}
 
These names are the machinery described in Sec.~\ref{sec:algorithm}. Ordinary
use of \MoMPy{} never needs them directly --- \code{Problem.build()} drives
them internally --- but they are public and documented here in full for
anyone extending the rewriting system or writing a new problem class.
 
\begin{description}[leftmargin=1.4em,style=nextline]
\item[\code{ZERO\_CLASS}]
  The disjoint-set id (\code{0}) reserved for the class of words that are
  identically zero; a sticky root (see \code{.union} below).
\item[\code{UnionFind()}]
  A disjoint-set forest with path halving and union by rank; \code{ZERO\_CLASS}
  is created as the initial root.
\item[\code{.new()}]
  Allocate and return a fresh singleton class id.
\item[\code{.find(x)}]
  The current representative (root) of \code{x}'s class, path-halving as it
  goes.
\item[\code{.union(a, b)}]
  Merge the classes of \code{a} and \code{b}; returns the surviving root.
  \code{ZERO\_CLASS} always wins regardless of rank, which is how a zero
  identification propagates through everything later merged into it, even
  merges discovered after the fact.
\item[\code{len(uf)}]
  Number of classes ever allocated (including ones since merged away).
\item[\code{Rewriter(algebra, *, tracial=True, hermitian=True)}]
  Generates the one-step moves available on a word (Sec.~\ref{sec:algorithm}).
  \code{tracial}: cyclic rotations are equivalent and the first/last letters
  count as adjacent; set \code{False} for block moment matrices.
  \code{hermitian}: a word is identified with its reversal.
\item[\code{.is\_zero(w)}]
  \code{True} iff two letters adjacent in \code{w} (including the wrap-around
  pair when tracial) are distinct members of one declared orthogonal set.
\item[\code{.neighbours(w)}]
  Generator yielding every word reachable from \code{w} by exactly one move: a
  single rotation (tracial mode; sufficient to expose the wrapped pair without
  enumerating every rotation), the reversal (hermitian mode), an idempotent
  collapse of an adjacent equal pair, or a swap of an adjacent commuting pair.
\item[\code{.reduce(w)}]
  Repeatedly collapses adjacent idempotent-equal pairs until none remain: a
  cheap, non-canonical pre-normalisation, not a substitute for the full search
  in \code{Classifier}.
\item[\code{Classifier(rewriter)}]
  Assigns words to equivalence classes, memoising as it goes; the public entry
  point is \code{.classify}.
\item[\code{.classify(word)}]
  Returns the (possibly not-yet-final) class id of \code{word}, exploring it by
  breadth-first search if new, and merging with every already-classified word
  the search touches.
\item[\code{.resolve(class\_id)}]
  Maps a possibly-stale class id to its current union-find representative;
  always call this before comparing two class ids for equality.
\item[\code{.known(word)}, \code{.lookup(word)}]
  \code{True} iff \code{word} has been classified; and the current class of an
  already-seen word (\code{None} if never seen, without triggering a search).
\item[\code{.mark\_zero(word)}]
  Force \code{word}'s whole class to merge into the zero class; returns the
  resulting root.
\item[\code{.classes()}]
  \code{\{root: [words...]\}} for the current, fully-merged state.
\item[\code{.zero\_words()}, \code{.all\_words()}]
  The words in the current zero class; and an iterable over every word ever
  classified.
\item[\code{len(classifier)}, \code{.words\_expanded}]
  Number of distinct words classified so far; and how many were actually
  expanded by the breadth-first search (reported in \code{MomentMatrix.stats}
  as \code{words\_expanded}).
\end{description}
 
\subsection{Numerical block-diagonalisation --- \texttt{MoMPy.blkdiag}}
\label{app:blkdiag}
 
This module is not imported by \code{MoMPy/\_\_init\_\_.py} and so is not part
of \code{from MoMPy import *}; use \code{from MoMPy.blkdiag import blkdiag,
M\_ones} explicitly. It also assumes \code{numpy} has already been imported as
\code{np} in the caller's session rather than importing it itself, so add
\code{import numpy as np} before use.
 
\begin{description}[leftmargin=1.4em,style=nextline]
\item[\code{M\_ones(matrix, element, eps)}]
  \code{matrix}: a square numerical matrix. \code{element}: a numerical value
  to search for. \code{eps}: tolerance; entries within \code{eps} of
  \code{element} count as matches. Returns a same-shape $0/1$ matrix,
  symmetrised, marking every position whose entry in \code{matrix} is
  \code{eps}-close to \code{element}. Used by \code{blkdiag} as the indicator
  matrix for one recurring value.
\item[\code{blkdiag(MomMat, id\_els)}]
  \code{MomMat}: a numerical moment matrix, e.g.\ an explicit numerical
  realisation of a built \code{MomentMatrix}. \code{id\_els}: the distinct
  numerical values whose recurring positions should be block-diagonalised
  together (typically the distinct SDP-variable values). Builds one indicator
  matrix per value via \code{M\_ones}, forms two random self-adjoint
  combinations of them, diagonalises the first to get a candidate eigenbasis,
  and uses the second to sort basis vectors into common blocks by their
  nonzero couplings --- the randomised simultaneous block-diagonalisation used
  for symmetry reduction in semidefinite
  programming~\cite{Rosset2018,IoannouRosset2021}. \MoMPy{} has no automatic
  symmetry-reduction pass of its own: this utility only supplies the
  numerical linear algebra once a symmetry group is already in hand; finding
  the right \code{id\_els} from that group is left to the user. Returns \code{(Trans, ListBlocks, P)}: the
  orthogonal transformation that block-diagonalises, the list of block sizes
  in the transformed basis, and each indicator matrix expressed in that basis.
\end{description}
 
\section{Legacy interface}
\label{app:legacy}
 
Scripts written against \MoMPy~0.x continue to run unchanged, through two
modules that wrap the modern engine: \texttt{MoMPy.MoM} (tracial matrices) and
\texttt{MoMPy.BloM} (block matrices).
 
\begin{lstlisting}
from MoMPy.MoM import MomentMatrix, fmap, normalisation_contraints
[G, map_table, S, eq_indices, Gexp] = MomentMatrix(
    S_1, S_2, S_high, rank_1, orthogonal_projectors, commuting_pairs)
\end{lstlisting}
 
\code{map\_table} is now a \code{list} subclass carrying a hash index, so
\code{map\_table[-1][-1]} and \code{term[0]} behave as before while
\code{fmap} became a single dictionary probe. Note that the current interface fixes several
defects in the legacy 0.x one, two of which change numerical results --- the moment matrix and
the lookup table now always agree, and equivalence classes are closed
completely. Both changes tighten the relaxation. See \code{MIGRATION.md} in the
repository for details.
 
\subsection{\texttt{MoMPy.MoM}: tracial matrices}
 
\begin{description}[leftmargin=1.6em,style=nextline]
\item[\code{MomentMatrix(S\_1, S\_2, higher\_order\_elements, rank\_1\_projectors, orthogonal\_projectors, commuting\_pairs, progress=False)}]
  The 0.x functional constructor (module-level; \emph{not} the same object as
  the \code{MomentMatrix} class of App.~\ref{app:matrix}, though it shares the
  name --- import it as \code{MoMPy.MoM.MomentMatrix}). \code{S\_1}:
  first-order labels (identity added automatically). \code{S\_2}: labels whose
  ordered pairs (all $h,k\in$\code{S\_2}, including repeats) form the
  second-order monomials; \code{[]} to skip. \code{higher\_order\_elements}:
  explicit monomials of any length, each a list of labels.
  \code{rank\_1\_projectors}, \code{orthogonal\_projectors},
  \code{commuting\_pairs}: the relations, in the pre-\code{Algebra} shape.
  \code{progress}: print a progress line while building. Returns the legacy
  5-tuple \code{(Moment\_Matrix, map\_table, S, list\_of\_eq\_indices, Mexp)}.
  Internally builds an \code{Algebra} and a \code{MomentProblem} and delegates
  to the modern engine via \code{.to\_legacy()}.
\item[\code{fmap(table, i)}]
  Maps monomial \code{i} to its SDP variable index in \code{table}. A single
  dict probe against a \code{MapTable}; a linear scan for a hand-built plain
  list of \code{[members, index]} rows. Returns the literal \code{FMAP\_ERROR}
  string (not an exception) on a miss, so \code{if fmap(...) == 'ERROR: ...'}
  guard code keeps working.
\item[\code{normalisation\_contraints(element, list\_identities\_in)}]
  \emph{(Legacy spelling, missing an ``s''.)} \code{element}: a POVM's outcome
  labels. \code{list\_identities\_in}: the map-table-shaped list of
  equivalence groups to search. Returns a list of blocks, each of length
  \code{len(element)+1}: one monomial per outcome substituted into a shared
  position, then the monomial with that position deleted. Scans every word of
  every group and deduplicates by substitution site, so --- unlike the original 0.x
  behaviour --- it returns a complete, non-redundant set regardless of which
  word a class happened to store first.
\item[\code{normalisation\_contraints\_2compatibility(Belement, Melement, list\_identities\_in)}]
  The marginalisation/joint-measurability analogue: \code{Belement} are a
  joint POVM's outcome labels, \code{Melement} the single operator they
  marginalise onto. Returns blocks
  \code{[w\_with\_B0, ..., w\_with\_Bk, w\_with\_M]}. Superseded by
  \code{MoMPy.constraints.marginal\_constraints} (App.~\ref{app:constraints}),
  its corrected, deduplicated successor.
\item[\code{check\_if\_id(element, map\_table, rank\_1\_projectors, commuting\_elements, orthogonal\_projectors)}]
  Checks whether \code{element} already occurs in \code{map\_table}; the
  trailing three arguments are accepted for signature compatibility but not
  used by the lookup. Returns \code{[found, is\_zero, index]}. Superseded by
  \code{MapTable.get(word)}, a direct single-probe replacement.
\item[\code{Permute(v)}]
  Cyclic right-rotation: \code{[a,b,c]} $\to$ \code{[c,a,b]} (mirrors
  $\Tr(ABC)=\Tr(CAB)$).
\item[\code{reverse\_list(lista)}]
  Returns \code{list(reversed(lista))}.
\item[\code{Commute\_new(vec, i, j)}]
  Returns a copy of \code{vec} with positions \code{i} and \code{j} swapped;
  raises \code{IndexError} if either is out of range.
\item[\code{Commute(v, index)}]
  Returns a copy of \code{v} with \code{v[index]} swapped with the
  \emph{next} entry, wrapping cyclically (\code{index} taken mod
  \code{len(v)}). Corrected relative to 0.x, which removed elements
  \emph{by value} and so misbehaved on words with repeated labels; this
  version swaps by index.
\end{description}
 
\subsection{\texttt{MoMPy.BloM}: block matrices}
 
Re-exports \code{fmap}, \code{Permute}, \code{Commute}, \code{Commute\_new},
\code{reverse\_list} and \code{normalisation\_contraints\_2compatibility}
unchanged from \code{MoMPy.MoM} (above), plus:
 
\begin{description}[leftmargin=1.6em,style=nextline]
\item[\code{BlockMatrix(S\_1, S\_2, higher\_order\_elements, rank\_1\_projectors, orthogonal\_projectors, commuting\_pairs, progress=False)}]
  Same five positional arguments and return shape as \code{MoM.MomentMatrix}
  above, but builds a block moment matrix: cyclicity is not imposed, so
  entries are identified only up to the declared commutation/idempotency
  relations. Internally builds the unified \code{MomentProblem} with
  \code{dim=1, cyclicity=False, hermitian=False, dedupe=False}.
\item[\code{block\_normalisation\_contraints(element, list\_identities\_in)}]
  The block-matrix counterpart of \code{MoM.normalisation\_contraints}; the
  implementation is shared verbatim (normalisation does not care whether
  entries are traces or operators), so behaviour and return shape are
  identical to that function's.
\item[\code{check\_if\_id\_BloM(element, map\_table, rank\_1\_projectors=None, commuting\_elements=None, orthogonal\_projectors=None)}]
  Same purpose and return shape as \code{MoM.check\_if\_id}
  (\code{[found, is\_zero, index]}), reimplemented as a table lookup; the
  trailing keyword arguments are accepted and ignored. In 0.x this function
  referenced a \code{commuting\_pairs} name that was never a parameter and
  raised \code{NameError} on every call; this is the corrected version.
\end{description}
\end{sloppypar}
 
\end{document}